\documentclass[11pt,british,english,ruled,linesnumbered,vlined]{article}
\usepackage[T1]{fontenc}
\usepackage[latin9]{inputenc}
\usepackage{color}
\usepackage{babel}
\usepackage{float}
\usepackage{textcomp}
\usepackage{mathrsfs}
\usepackage{algorithm2e}
\usepackage{amsmath}
\usepackage{amsthm}
\usepackage{amssymb}
\usepackage{graphicx}
\usepackage{esint}
\usepackage[authoryear]{natbib}
\usepackage[unicode=true,pdfusetitle,
 bookmarks=true,bookmarksnumbered=false,bookmarksopen=false,
 breaklinks=false,pdfborder={0 0 1},backref=false,colorlinks=true]
 {hyperref}

\makeatletter

\providecommand{\tabularnewline}{\\}

  \theoremstyle{remark}
  \newtheorem{rem}{\protect\remarkname}
\theoremstyle{plain}
\newtheorem{thm}{\protect\theoremname}
  \theoremstyle{plain}
  \newtheorem{lem}{\protect\lemmaname}
  \theoremstyle{plain}
  \newtheorem{cor}{\protect\corollaryname}
  \theoremstyle{plain}
  \newtheorem{prop}{\protect\propositionname}

\usepackage{babel}
\usepackage{mathrsfs}
\usepackage{breakurl}
\usepackage{multirow}

 \theoremstyle{definition}
  \theoremstyle{remark}

\usepackage{amsfonts}
\usepackage{color}
\usepackage[usenames, dvipsnames]{xcolor}
\usepackage{fullpage}
\usepackage{subfigure}
\usepackage{comment}
\usepackage{epsfig}\usepackage{authblk}
\usepackage{bbm}

\usepackage[font=footnotesize, labelfont=bf]{caption}

\graphicspath{ ./ } 

\providecommand{\propositionname}{Proposition}
\providecommand{\remarkname}{Remark}
\providecommand{\theoremname}{Theorem}

\providecommand{\remarkname}{Remark}

\newcounter{hypA}

\newenvironment{condition}{\refstepcounter{hypA}\begin{itemize}
\item[({\bf A\arabic{hypA}})]}{\end{itemize}}

  \addto\captionsbritish{\renewcommand{\lemmaname}{Lemma}}
  \addto\captionsbritish{\renewcommand{\propositionname}{Proposition}}
  \addto\captionsbritish{\renewcommand{\remarkname}{Remark}}
  \addto\captionsenglish{\renewcommand{\lemmaname}{Lemma}}
  \addto\captionsenglish{\renewcommand{\propositionname}{Proposition}}
  \addto\captionsenglish{\renewcommand{\remarkname}{Remark}}
  \providecommand{\lemmaname}{Lemma}
  \providecommand{\propositionname}{Proposition}
  \providecommand{\remarkname}{Remark}
\addto\captionsbritish{\renewcommand{\corollaryname}{Corollary}}
\addto\captionsbritish{\renewcommand{\theoremname}{Theorem}}
\addto\captionsenglish{\renewcommand{\corollaryname}{Corollary}}
\addto\captionsenglish{\renewcommand{\theoremname}{Theorem}}
\providecommand{\corollaryname}{Corollary}
\providecommand{\theoremname}{Theorem}

\makeatother

  \addto\captionsbritish{\renewcommand{\lemmaname}{Lemma}}
  \addto\captionsbritish{\renewcommand{\propositionname}{Proposition}}
  \addto\captionsbritish{\renewcommand{\remarkname}{Remark}}
  \addto\captionsenglish{\renewcommand{\lemmaname}{Lemma}}
  \addto\captionsenglish{\renewcommand{\propositionname}{Proposition}}
  \addto\captionsenglish{\renewcommand{\remarkname}{Remark}}
  \providecommand{\lemmaname}{Lemma}
  \providecommand{\propositionname}{Proposition}
  \providecommand{\remarkname}{Remark}
\addto\captionsbritish{\renewcommand{\corollaryname}{Corollary}}
\addto\captionsbritish{\renewcommand{\theoremname}{Theorem}}
\addto\captionsenglish{\renewcommand{\corollaryname}{Corollary}}
\addto\captionsenglish{\renewcommand{\theoremname}{Theorem}}
\providecommand{\corollaryname}{Corollary}
\providecommand{\theoremname}{Theorem}

\begin{document}

\title{Scalable Monte Carlo inference for state-space models}

\author{Sinan Y\i ld\i r\i m$^{*}$, Christophe Andrieu$^{\dagger}$ and
Arnaud Doucet$^{\ddagger}$ }

\maketitle
$^{*}$Faculty of Engineering and Natural Sciences, Sabanc\i{} University,
Turkey. 

\medskip{}

$\dagger$School of Mathematics, Bristol University, UK.

\medskip{}

$\ddagger$Department of Statistics, Oxford University, UK.

\medskip{}
\begin{abstract}
We present an original simulation-based method to estimate likelihood
ratios efficiently for general state-space models. Our method relies
on a novel use of the conditional Sequential Monte Carlo (cSMC) algorithm
introduced in \citet{Andrieu_et_al_2010} and presents several practical
advantages over standard approaches. The ratio is estimated using
a unique source of randomness instead of estimating separately the
two likelihood terms involved. Beyond the benefits in terms of variance
reduction one may expect in general from this type of approach, an
important point here is that the variance of this estimator decreases
as the distance between the likelihood parameters decreases. We show
how this can be exploited in the context of Monte Carlo Markov chain
(MCMC) algorithms, leading to the development of a new class of exact-approximate
MCMC methods to perform Bayesian static parameter inference in state-space
models. We show through simulations that, in contrast to the Particle
Marginal Metropolis\textendash Hastings (PMMH) algorithm of \citet{Andrieu_et_al_2010},
the computational effort required by this novel MCMC scheme scales
very favourably for large data sets.

\end{abstract}

\section{Introduction\label{sec: Introduction}}

State-space models (SMMs) form an important class of statistical model
used in many fields; see \citet{Douc_et_al_2014} for a recent overview.
In its simplest form a SSM is comprised of an $(\mathsf{X},\mathcal{X})$-valued
latent Markov chain $\{X_{t};t\geq1\},$ and a $(\mathsf{Y},\mathcal{Y})$-valued
observed process $\{Y_{t};t\geq1\}$. The latent process has initial
probability with density $\eta_{\theta}(x_{1})$ and transition density
$f_{\theta}(x_{t-1},x_{t})$; both probability densities defined on
$\mathsf{X}$ and with respect to a common dominating measure on $(\mathsf{X},\mathcal{X})$
denoted generically as ${\rm d}x$ and parametrized by some $\theta\in\Theta\subset\mathbb{R}^{d_{\theta}}$.
Naturally, non-dynamical models for which $f_{\theta}(x_{t-1},x_{t})=f_{\theta}(x_{t})$
form a particular case. The observation at time $t$ is assumed conditionally
independent of all other random variables given $X_{t}=x_{t}$ and
its conditional observation density is $g_{\theta}(x_{t},y_{t})$
on $\mathsf{Y}$ with respect to the dominating measure ${\rm d}y$
on $(\mathsf{Y},\mathcal{Y})$. For a given value $\theta\in\Theta$
we will refer to this model as $\mathscr{M}_{\theta}$, and the corresponding
joint density of the latent and observed variables up to time $T$
is 
\begin{equation}
p_{\theta}(x_{1:T},y_{1:T})=\mu_{\theta}(x_{1})\prod_{t=2}^{T}f_{\theta}(x_{t-1},x_{t})\prod_{t=1}^{T}g_{\theta}(x_{t},y_{t})\;.\label{eq: HMM joint density}
\end{equation}
from which the likelihood function associated to the observations
$y_{1:T}$ can be obtained 
\begin{equation}
l_{\theta}(y_{1:T}):=\int_{\mathsf{X}^{T}}p_{\theta}(x_{1:T},y_{1:T}){\rm d}x_{1:T}\;.\label{eq:likelihoodSSM}
\end{equation}
Such models are typically intractable, therefore requiring the use
of numerical methods to carry out inference about $\theta,x_{1:T}$.
Significant progress was made in the 1990s and early 2000's to solve
numerically the so-called filtering/smoothing problem, that is, assuming
$\theta\in\Theta$ known, efficient methods were proposed to approximate
the posterior density $\pi_{\theta}(x_{1:T}):=p_{\theta}(x_{1:T}\mid y_{1:T})$
or some of its marginals. Indeed particle filters, or more generally
Sequential Monte Carlo methods (SMC), have been shown to provide a
set of versatile and efficient tools to approximate the aforementioned
posteriors by exploiting the sequential structure of $p_{\theta}(x_{1:T}\mid y_{1:T})$,
and their theoretical properties are now well understood \citep{Del_Moral_2004}.

Estimating $\theta\in\Theta$, the static parameter, is however known
to be much more challenging. Indeed, likelihood based methods (e.g.
maximum likelihood or Bayesian estimation) usually require evaluation
of $l_{\theta}(y_{1:T})$ or its derivatives in order to be implemented
practically; see \citet{Kantas_et_al_2015} for a recent review. As
we shall see, of particular interest is the estimation of the likelihood
ratio, that is for $\theta,\theta'\in\Theta$, 
\[
\mathfrak{L}(\theta,\theta'):=\frac{l_{\theta'}(y_{1:T})}{l_{\theta}(y_{1:T})}\;.
\]
In a classical set-up $\mathfrak{L}(\theta,\theta')$ plays a central
role in testing, for example, but is also a direct route to the numerical
evaluation of the gradient of the log-likelihood function or the implementation
of Markov chain Monte Carlo (MCMC) algorithms used to perform Bayesian
inference.

The first contribution of the present paper is the realization that
the conditional SMC (cSMC) kernel introduced in \citet{Andrieu_et_al_2010},
an MCMC kernel to sample from $\pi_{\theta}\big({\rm d}x_{1:T}\big)$,
can be combined with Annealing Importance Sampling (AIS) \citep{Crooks1998,Neal_2001}
in order to develop efficient estimators of $\mathfrak{L}(\theta,\theta')$.
Central to the good behaviour of this class of estimators is the fact
that rather than estimating numerator and denominator independently,
as suggested by current methods, this is here performed jointly using
a unique source of randomness. Alternative approaches exploiting this
principle have been explored briefly in \citet{Lee_Holmes_2010} and
studied more thoroughly in \citet{delagiannidis:doucet:pitt:2015}
in the context of MCMC simulations. Our estimator differs substantially
from these earlier proposals. The second contribution here is to provide
theory for this novel likelihood ratio estimator and show how this
estimator can be exploited in numerical procedures in order to design
algorithms which scale well with the number of data points. In particular
we present a new exact approximate MCMC scheme for perform Bayesian
static parameter inference in SSMs and we demonstrate its performance
through simulations. 

\section{Likelihood ratio estimation in SSM with cSMC\label{sec: cSMC-AIS}}

An efficient technique to estimate $l_{\theta}(y_{1:T})$ for $\theta\in\Theta$
consists of using SMC methods. The algorithm is presented in Algorithm
\ref{alg:SMC}; it requires a user defined instrumental probability
distribution $m_{\theta}(\cdot):\mathcal{X}\rightarrow[0,1]$ and
a Markov kernel $M_{\theta}(\cdot,\cdot):\mathsf{X}\times\mathcal{X}\rightarrow[0,1]$,
referred to as $\mathscr{A}_{\theta}=\{m_{\theta},M_{\theta}\}$\textendash $M_{\theta}(\cdot,\cdot)$
can be made time dependent, but we aim to keep notation simple here.
We also use the notation $\mathcal{P}\big(\omega^{(1)},\omega^{(2)},\ldots,\omega^{(N)}\big)$
to refer to the probability distribution of a discrete valued random
variable $B$ taking values in $\{1,2,\ldots,N\}$ such that $\mathbb{P}\big(B=b\big)\propto\omega^{(b)}$.
An estimator of the likelihood can be obtained by 
\begin{equation}
\hat{l}_{\theta}(y_{1:T}):=\prod_{t=1}^{T}\frac{1}{N}\sum_{i=1}^{N}w_{t}^{(i)}\;.\label{eq:estimate_ell_SMC}
\end{equation}
This estimator has attractive properties. It is unbiased \citep{Del_Moral_2004}
and has a relative variance which scales linearly in $T$ under practically
relevant conditions \citep{cerou2011nonasymptotic}. One can therefore
use two such independent SMC estimators for $\theta$ and $\theta'$
and compute their ratio to estimate $\mathfrak{L}(\theta,\theta')$.
However better estimators are possible if one introduces positive
dependence between the two estimators, this is exploited in \citet{delagiannidis:doucet:pitt:2015}
and \citet{Lee_Holmes_2010}. Our approach relies on the same idea
but the estimator we propose is very different from these earlier
proposals and complementary, as discussed later in the paper.

\begin{algorithm}
\For{ $i=1,\ldots,N$}{

Sample $z_{1}^{(i)}\sim m_{\theta}(\cdot)$

Compute $w_{1}^{(i)}=\mu_{\theta}\big(z_{1}^{(i)}\big)g_{\theta}\big(z_{1}^{(i)},y_{1}\big)/m_{\theta}\big(z_{1}^{(i)}\big)$

}

\For{ $t=2,\ldots,T$}{

\For{ $i=1,\ldots,N$}{

Sample $a_{t-1}^{(i)}\sim\mathcal{P}\big(w_{t-1}^{(1)},\ldots,w_{t-1}^{(N)}\big)$
and $z_{t}^{(i)}\sim M_{\theta}\big(z_{t-1}^{(a_{t-1}^{(i)})},\cdot\big)$

Compute $w_{t}^{i}=f_{\theta}\big(z_{t-1}^{(a_{t-1}^{(i)})},z_{t}^{(i)}\big)g_{\theta}\big(z_{t}^{(i)},y_{t}\big)/M_{\theta}\big(z_{t-1}^{(a_{t-1}^{(i)})},z_{t}^{(i)}\big)$

}

}

\Return$\big(a_{1:T}^{1:N},z_{1:T}^{1:N},w_{1:T}^{1:N}\big)$

\protect\caption{\label{alg:SMC}$\mathrm{SMC}\big(N,\mbox{\ensuremath{\mathscr{M_{\theta},\mathscr{A}_{\theta}}}}\big)$}
\end{algorithm}

We rely here on the AIS method of \citet{Crooks1998,Neal_2001} which
is a state-of-the-art Monte Carlo approach to estimate ratio of normalizing
constants. For $\theta,\theta'\in\Theta$ the method requires one
to first choose a family of probability distributions $\mathscr{P}_{\theta,\theta'}=\big\{\pi_{\theta,\theta',\varsigma},\varsigma\in[0,1]\big\}$
defined on $(\mathsf{X}^{T},\mathcal{X}^{\otimes T})$, whose aim
is to ``bridge'' $\pi_{\theta}$ and $\pi_{\theta'}$, and a family
of transition probabilities $\mathscr{R}_{\theta,\theta'}=\big\{ R_{\theta,\theta',\varsigma}(\cdot,\cdot):\mathsf{X}^{T}\times\mathcal{X}^{\otimes T}\to[0,1],\varsigma\in[0,1]\big\}$
and a mapping $\varsigma(\cdot):[0,1]\rightarrow[0,1]$. The role
of these quantities is clarified below. For $\theta,\theta'\in\Theta$
we say that $\mathscr{P}_{\theta,\theta'}$, $\mathscr{R}_{\theta,\theta'}$
and $\varsigma(\cdot)$ associated with $\mathscr{M}_{\theta}$ and
$\mathscr{M}_{\theta'}$ satisfy (A\ref{hypA:basicconditions}) if

\begin{condition}\label{hypA:basicconditions}Conditions on $\mathscr{P}_{\theta,\theta'}$,
$\mathscr{R}_{\theta,\theta'}$ and $\varsigma(\cdot)$, 
\begin{enumerate}
\item $\mathscr{P}_{\theta,\theta'}=\big\{\pi_{\theta,\theta',\varsigma},\varsigma\in[0,1]\big\}$
is a family of probability distributions on $(\mathsf{X}^{T},\mathcal{X}^{\otimes T})$
satisfying

\begin{enumerate}
\item the end point conditions $\pi_{\theta,\theta',0}(\cdot)=\pi_{\theta}(\cdot)$
and $\pi_{\theta,\theta',1}(\cdot)=\pi_{\theta'}(\cdot)$ as defined
by $\mathscr{M}_{\theta}$ and $\mathscr{M}_{\theta'}$, 
\item for any $A\in\mathcal{X}^{\otimes T}$ and $\varsigma,\varsigma'\in[0,1]$
such that $\varsigma\leq\varsigma'$, $\pi_{\theta,\theta',\varsigma'}(A)>0$
implies $\pi_{\theta,\theta',\varsigma}(A)>0$, 
\end{enumerate}
\item $\mathscr{R}_{\theta,\theta'}=\big\{ R_{\theta,\theta',\varsigma}(\cdot,\cdot):\mathsf{X}^{T}\times\mathcal{X}^{\otimes T}\to[0,1],\varsigma\in[0,1]\big\}$
is such that for any $\varsigma\in[0,1]$, $R_{\theta,\theta',\varsigma}(\cdot,\cdot)$
leaves $\pi_{\theta,\theta',\varsigma}(\cdot)$ invariant, 
\item $\varsigma(\cdot):[0,1]\rightarrow[0,1]$ a non-decreasing mapping
such that $\varsigma(0)=0$ and $\varsigma(1)=1$. 
\end{enumerate}
\end{condition}

In order to implement the AIS procedure, one chooses $K\in\mathbb{N}$
and considers the sub-family of probability distributions $\mathscr{P}_{\theta,\theta',K}:=\{\pi_{\theta,\theta',k}^{[K]},k=0,\ldots,K+1\}\subset\mathscr{P}_{\theta,\theta'}$
such that for any $k=0,\ldots,K+1$ $\pi_{\theta,\theta',k}^{[K]}=\pi_{\theta,\theta',\varsigma(k/(K+1))}$
and the corresponding family of transition kernels $\mathscr{R}_{\theta,\theta',K}:=\big\{ R_{\theta,\theta',k}^{[K]}(\cdot,\cdot):\mathsf{X}^{T}\times\mathcal{X}^{\otimes T}\to[0,1],k=1,\ldots,K\big\}$.
The integer $K$ therefore represents the number of intermediate distributions
introduced to bridge $\pi_{\theta}(\cdot)$ and $\pi_{\theta'}(\cdot)$,
which is allowed to be zero. For notational simplicity we will drop
the dependence on $K$ of the elements of $\mathscr{P}_{\theta,\theta',K}$
and $\mathscr{R}_{\theta,\theta',K}$ when no ambiguity is possible.
Let $\mathbf{u}:=x_{1:T}$ and consider the non-homogeneous Markov
chain $\big\{\mathbf{U}_{i},i=0,\ldots,K\big\}$ such that $\mathbf{U}_{0}\sim\pi_{\theta}$
and for $k\geq1$ $\mathbf{U}_{k}\mid\mathbf{U}_{k-1}=\mathbf{u}_{k-1})\sim R_{\theta,\theta',k}\big(\mathbf{u}_{k-1},\cdot\big)$.
It is routine to show that under these assumptions the quantity 
\[
\prod_{k=0}^{K}\frac{\pi_{\theta,\theta',k+1}\big(\mathbf{U}_{k}\big)}{\pi_{\theta,\theta',k}\big(\mathbf{U}_{k}\big)}
\]
has expectation $1$. The interest of this identity is that whenever
$\pi_{\theta,\theta',\varsigma}=\gamma_{\theta,\theta',\varsigma}/Z_{\theta,\theta',\varsigma}$
where $Z_{\theta,\theta',\varsigma}$ is an unknown normalising constant
but $\gamma_{\theta,\theta',\varsigma}$ can be evaluated pointwise
then

\begin{equation}
\prod_{k=0}^{K}\frac{\gamma_{\theta,\theta',k+1}(\mathbf{U}_{k})}{\gamma_{\theta,\theta',k}(\mathbf{U}_{k})}\label{eq:AISestimator}
\end{equation}
is an unbiased estimator of $Z_{\theta,\theta',K+1}/Z_{\theta,\theta',0}$.
Consequently, for $\gamma_{\theta,\theta',0}(x_{1:T})=p_{\theta}(x_{1:T},y_{1:T})$
and $\gamma_{\theta,\theta',K+1}(x_{1:T})=p_{\theta'}(x_{1:T},y_{1:T})$,
this provides us with a way of estimating the desired likelihood ratio
$\mathcal{L}(\theta,\theta')$. The algorithm is summarized in Algorithm
\ref{alg:Annealing-Importance-Sampling}, which should be initialised
with $x_{1:T}\sim\pi_{\theta}(\cdot)$ to lead to an unbiased estimator
of $\mathcal{L}(\theta,\theta')$.

\begin{algorithm}
Set $\mathbf{u}_{0}=x_{1:T}$

\For{ $k=1,\ldots,K$,}{

Sample $\mathbf{u}_{k}\sim R_{\theta,\theta',k}(\mathbf{u}_{k-1},\cdot)$
targetting $\pi_{\theta,\theta',k}$

}

Compute

\begin{equation}
\mathfrak{\hat{L}}(\theta,\theta')=\prod_{k=0}^{K}\frac{\gamma_{\theta,\theta',k+1}(\mathbf{u}_{k})}{\gamma_{\theta,\theta',k}(\mathbf{u}_{k})}\label{eq: noisy likelihood ratio}
\end{equation}

\Return$\big(\mathfrak{\hat{L}}(\theta,\theta'),\mathbf{u}_{K}\big)$

\protect\caption{\label{alg:Annealing-Importance-Sampling}$\mathrm{AIS}\big(x_{1:T},\mathscr{P}_{\theta,\theta'},\mathscr{R}_{\theta,\theta'},K,\varsigma(\cdot)\big)\;.$}
\end{algorithm}

In general sampling exactly from $\pi_{\theta}(\cdot)$ is not possible.
Instead one can run an MCMC with transition kernel $R_{\theta,\theta',0}$,
hence targetting $\pi_{\theta}$, for $P$ iterations. Provided $R_{\theta,\theta',0}$
is ergodic one can feed $x_{1:T}\sim R_{\theta,\theta',0}^{P}$ into
$\mathrm{AIS}\big(x_{1:T},\mathscr{P}_{\theta,\theta'},\mathscr{R}_{\theta,\theta'},K,\varsigma(\cdot)\big)$
and control bias through $P$. There are several ways one can reduce
variability of this estimator. Under natural smoothness assumptions
on $\varsigma\mapsto\pi_{\theta,\theta',\varsigma},R_{\theta,\theta',\varsigma}$
and the mapping $\varsigma(\cdot)$, and ergodicity of $R_{\theta,\theta',\varsigma}$
one can show that this estimator is consistent as $K\rightarrow\infty$.
More simply it is also possible, for $K$ fixed, to consider $M$
independent copies of the estimator and consider their average\textendash the
latter strategy has the advantage that it lends itself trivially to
parallel computing architectures, in contrast to the former.

There is an additional natural and ``free'' control of both bias
and variance when computation of $\mathfrak{\hat{L}}(\theta,\theta')$
is required only for $\theta$ and $\theta'$ ``close''. Indeed
in such scenarios, provided the models considered are smooth enough
in $\theta$, one expects the estimation of $\mathfrak{L}(\theta,\theta')$
to be easier since the densities $\pi_{\theta}(x_{1:T})$ and $\pi_{\theta'}(x_{1:T})$
will be close to one another. For illustration and concreteness we
briefly describe this fact in the context of a stochastic gradient
algorithm to maximize $l_{\theta}(y_{1:T})$\textendash the main focus
of the paper is on sampling, but this requires additional technicalities.
Assume $\nabla_{\theta}\log l_{\theta}(y_{1:T})$ is intractable and
that we wish to use a finite difference method to approximate this
quantity. The simultaneous perturbation (SPSA) approach of \citet{spall1992multivariate}
is such a method, which naturally lends itself to the use of our class
of estimators . Let $\delta$ be a, possibly random, $d_{\theta}-$dimensional
vector such that $\theta\pm\delta\in\Theta$, then a possible estimator
of $\nabla\log l_{\theta}(y_{1:T})$ could be the vector whose $i-$th
component is 
\[
\frac{1}{2[\delta]_{i}}\log\left(\frac{l_{\theta+\delta}(y_{1:T})}{l_{\theta-\delta}(y_{1:T})}\right)\approx\frac{1}{2[\delta]_{i}}\left(\frac{l_{\theta+\delta}(y_{1:T})}{l_{\theta-\delta}(y_{1:T})}-1\right)\;,
\]
which depends on the likelihood ratio $\mathcal{L}(\theta+\delta,\theta-\delta)$.
A natural idea is to plug-in the AIS estimator $\hat{\mathcal{L}}(\theta+\delta,\theta-\delta)$
developed earlier and note that such a strategy is likely to be better
than a strategy which would consists of estimating numerator and denominator
independently.

We now discuss natural choices of $\mathscr{P}_{\theta,\theta'}$
and $\mathscr{R}_{\theta,\theta'}$ for this AIS procedure in the
context of state-space models. These choices are crucial to the good
performance of the algorithm.

\subsection*{Choice of $\mathscr{P}_{\theta,\theta'}$}

A standard choice consists of using geometric annealing, that is define
for $\varsigma\in[0,1]$, 
\[
\gamma_{\theta,\theta',\varsigma}(x_{1:T}):=\gamma_{\theta}(x_{1:T})^{1-\varsigma}\gamma_{\theta'}(x_{1:T})^{\varsigma}\;,
\]
and, for example, set $\varsigma(t)=t$ for $t\in[0,1]$. This can
be written in a form similar to that arising from a state-space model
\[
\gamma_{\theta,\theta',\varsigma}(x_{1:T})=\tilde{\mu}_{\theta,\theta',\varsigma}(x_{1})\prod_{t=2}^{T}\tilde{f}_{\theta,\theta',\varsigma}(x_{t-1},x_{t})\prod_{t=1}^{T}\tilde{g}_{\theta,\theta',\varsigma}(x_{t},y_{t})\;,
\]
where for $x,x'\in\mathsf{X}$ and $y\in\mathsf{Y}$, $\tilde{\mu}_{\theta,\theta',\varsigma}(x)\propto\mu_{\theta}(x)^{1-\varsigma}\mu_{\theta}(x)^{\varsigma}$,
$\tilde{f}_{\theta,\theta',\varsigma}(x,x')\propto f_{\theta}(x,x')^{1-\varsigma}f_{\theta'}(x,x')^{\varsigma}$
and $\tilde{g}_{\theta,\theta',\varsigma}(x,y)\propto g_{\theta}(x,y)^{1-\varsigma}g_{\theta'}(x,y)^{\varsigma}$.
This could at first sight be a good choice since the sequential structure
of the model crucial to the implementation of efficient sampling techniques
is preserved. However, except for very specific cases such as when
the densities involved belong to the exponential family, the normalising
constant of $\tilde{f}_{\theta,\theta',\varsigma}(x,\cdot)$ may be
intractable, while being dependent on $\theta,\theta'$ and $x$.
While this is not an issue for the computation of \ref{eq:AISestimator},
this may lead to complications when implementing sampling techniques
relying on SMC (see Algorithm \ref{alg:SMC} and Remark \ref{rmk:Note-the-limitation}).
A way around this problem consists of defining $\vartheta(\cdot):[0,1]\rightarrow\Theta$
such that $\vartheta(0)=\theta$ and $\vartheta(1)=\theta'$, and
\[
\gamma_{\theta,\theta',\varsigma}(x_{1:T})=\gamma_{\vartheta(\varsigma)}(x_{1:T})\;,
\]
which trivially admits the desired sequential structure and defines
a tractable model. For example when $\Theta$ is convex the choice
$\vartheta(\varsigma)=(1-\varsigma)\theta+\varsigma\theta'$ will
always work.

\subsection*{Choice of $\mathscr{R}_{\theta,\theta'}$}

The conditional SMC (cSMC) algorithm belongs to the class of particle
MCMC algorithms introduced in \citet{Andrieu_Doucet_Holenstein_2009,Andrieu_et_al_2010}.
It is an SMC based algorithm (see Algorithm \ref{alg:SMC}) particularly
well suited to sampling from distributions arising from models with
a sequential structure, similar to that of $\pi_{\theta}$ for any
$\theta\in\Theta$. More precisely, for $\theta\in\Theta$ the cSMC
targetting $\pi_{\theta}$ yields a Markov transition kernel of invariant
distribution $\pi_{\theta}$, therefore lending itself to being used
as an MCMC method. The cSMC update has been shown both empirically
and theoretically to possess good convergence properties\textendash see
\citet{andrieuleevihola2013,Chopin_and_Singh_2015,lindsten2015uniform}
for recent studies of its theoretical properties. In its original
form the algorithm, corresponding to $\mathrm{cSMC}\big(\mathtt{False},N,x_{1:T},\mbox{\ensuremath{\mathscr{M_{\theta},\mathscr{A}_{\theta}}}}\big)$
in Algorithm \ref{alg:Conditional-SMC}, may suffer from the so-called
path degeneracy, meaning that because of the successive resampling
steps involved the particle paths $x_{1:T}$ at time $T$ have few
distinct values $x_{k}$ for $k\ll T$, resulting in poor mixing of
the corresponding MCMC. The cSMC with backward resampling as suggested
by \citet{Whiteley_2010} overcomes this problem by enabling reselection
of ancestors; a closely related approach is the ancestor resampling
technique of \citet{Lindsten_et_al_2014}. This is described in the
second part of Algorithm \ref{alg:Conditional-SMC}, and corresponds
to $\mathrm{cSMC}\big(\mathtt{True},N,x_{1:T},\mbox{\ensuremath{\mathscr{M_{\theta},\mathscr{A}_{\theta}}}}\big)$.

\begin{algorithm}
Set $z_{t}^{(1)}=x_{t}$ for $t=1,...,T$

\For{ $i=2,\ldots,N$}{

Sample $z_{1}^{(i)}\sim m_{\theta}(\cdot)$

Compute $w_{1}^{(i)}=\mu_{\theta}\big(z_{1}^{(i)}\big)g_{\theta}\big(z_{1}^{(i)},y_{1}\big)/m_{\theta}\big(z_{1}^{(i)}\big)$

}

\For{ $t=2,\ldots,T$}{

\For{ $i=2,\ldots,N$}{

Sample $a_{t-1}^{(i)}\sim\mathcal{P}\big(w_{t-1}^{(1)},\ldots,w_{t-1}^{(N)}\big)$
and $z_{t}^{(i)}\sim M_{\theta}\big(z_{t-1}^{(a_{t-1}^{(i)})},\cdot\big)$

Compute $w_{t}^{i}=f_{\theta}\big(z_{t-1}^{(a_{t-1}^{(i)})},z_{t}^{(i)}\big)g_{\theta}\big(z_{t}^{(i)},y_{t}\big)/M_{\theta}\big(z_{t-1}^{(a_{t-1}^{(i)})},z_{t}^{(i)}\big)$

}

}

Sample $k_{T}\sim\mathcal{P}\big(w_{T}^{(1)},\ldots,w_{T}^{(N)}\big)$
and set $x'_{T}=z_{T}^{(k_{T})}$

\For{ $t=T-1,\ldots,1$}{

\If{ $\neg{\rm BS}$}{ $k_{t}=a_{t}^{(k_{t+1})}$}

\Else{

\For{ $i=1,\ldots,N$}{

Compute $\tilde{w}_{t}^{(i)}=w_{t}^{(i)}f_{\theta}\big(z_{t}^{(i)},z_{t+1}^{(k_{t+1})}\big)$

}

Sample $k_{t}\sim\mathcal{P}\big(\tilde{w}_{t}^{(1)},\ldots,\tilde{w}_{t}^{(N)}\big)$

}

Set $x'_{t}=z_{t}^{(k_{t})}$

}

\Return $x'_{1:T}$

\protect\caption{\label{alg:Conditional-SMC}$\mathrm{cSMC}\big(\mathtt{BS},N,x_{1:T},\mbox{\ensuremath{\mathscr{M_{\theta},\mathscr{A}_{\theta}}}}\big)$ }
\end{algorithm}

Reversibility of cSMC with or without backward sampling with respect
to $\pi_{\theta}(\cdot)$ as well as its theoretical superiority over
the original cSMC are proven in \citet{Chopin_and_Singh_2015}. As
shown in \citep{Chopin_and_Singh_2015,andrieuleevihola2013,lindsten2015uniform},
convergence to stationarity can be made arbitrarily fast as $N$ increases.
For conciseness we will refer to ${\rm AIS}$ in Algorithm \ref{alg:Annealing-Importance-Sampling}
for which $\mathscr{R}_{\theta,\theta'}$ consists of $\mathrm{cSMC}\big(\mathtt{True},N,x_{1:T},\mbox{\ensuremath{\mathscr{M_{\vartheta},\mathscr{A}_{\vartheta}}}}\big)$
for all relevant $\vartheta$'s as ${\rm AIS}-\mathrm{cSMC}\big(x_{1:T},\mathscr{P}_{\theta,\theta'},\mathscr{A}_{\theta,\theta'},N,K,\varsigma(\cdot)\big)$
where $\mathcal{A}_{\theta,\theta'}$ is the set of instrumental methods
$\mathcal{A}_{\vartheta}$ required to implement the cSMCs targetting
the distributions in $\mathscr{P}_{\theta,\theta'}$. 
\begin{rem}
\label{rmk:Note-the-limitation}Contrary to the original cSMC, cSMC
with backward sampling is limited to scenarios where the transition
density $f_{\theta}$ is computable pointwise. Even when pointwise
evaluation is feasible, the backward sampling approach will be inefficient
if $f_{\theta}$ is close to singular; e.g. if $f_{\theta}$ arises
from the fine time discretization of a diffusion process.
\end{rem}

\begin{rem}
It is clear that there is another way of reducing variability : one
can draw several paths in the backward sampling stage and average
the corresponding estimators. We do not pursue this here.
\end{rem}

\section{Application to exact approximate MCMC for SSM\label{sec:Application-to-EA}}

In a Bayesian framework, the static parameter is ascribed a probability
distribution with density $\eta(\theta)$ (with respect to a dominating
measure denoted ${\rm d}\theta$) from which one defines the posterior
distribution of $(\theta,x_{1:T})$ given observations $y_{1:T}$
with density 
\begin{equation}
\pi(\theta,x_{1:T})\propto\eta(\theta)p_{\theta}(x_{1:T},y_{1:T})\;,\label{eq:jointposterior}
\end{equation}
(we drop $y_{1:T}$ in $\pi(\cdot)$ for notational simplicity). This
posterior distribution and its marginal $\pi({\rm d}\theta)$ are
potentially highly complex objects to manipulate in practice and (sampling)
Monte Carlo methods are often the only viable methods available to
extract information from such models. Assume for a moment that our
primary interest is in inferring $\theta$, and therefore that sampling
from $\pi({\rm d}\theta)$ is our concern. Among Monte Carlo methods,
MCMC techniques are often the only possible option\textendash we however
refer the reader to \citet{crisan:miguez:2013arXiv,Kantas_et_al_2015}
for purely particle based on-line methods. MCMC rely on the design
of ergodic Markov chains with the distribution of interest as invariant
distribution, say $\{\theta^{(i)},i\geq0\}$ with invariant distribution
$\pi({\rm d}\theta)$ for our problem. The Metropolis\textendash Hastings
(MH) algorithm plays a central role in the design of MCMC transition
probabilities, and proceeds as follows in our context. Given a family
of user defined and instrumental probability distributions $\big\{ q(\theta,\cdot),\theta\in\Theta\big\}$
on $\Theta$,

\begin{algorithm}
\protect\caption{Marginal algorithm}

\label{alg: MH}Given the current state $\theta$

Sample $\theta'\sim q(\theta,\cdot)$

Set the next state to $\theta'$ with probability $\min\{1,r(\theta,\theta')\}$,
where 
\begin{equation}
r(\theta,\theta'):=\frac{q(\theta',\theta)\pi(\theta')}{q(\theta,\theta')\pi(\theta)}=\frac{q(\theta',\theta)\eta(\theta')}{q(\theta,\theta')\eta(\theta)}\frac{l_{\theta'}(y_{1:T})}{l_{\theta}(y_{1:T})}\label{eq:genericMHacceptratio}
\end{equation}
Otherwise set the next state to $\theta$. 
\end{algorithm}

We will refer to $r(\theta,\theta')$ as the acceptance ratio and
call this MH algorithm targeting $\pi(\theta)$ the marginal MH algorithm.
A crucial point for the implementation of the algorithm is the requirement
to be able to evaluate the likelihood ratio $\mathfrak{L}(\theta,\theta')$.
This significantly reduces the class of models for which the algorithm
above can be used. In particular, one cannot apply this algorithm
to non-linear non-Gaussian SSMs as the likelihood (\ref{eq:likelihoodSSM})
is intractable. 

\subsection{State of the art}

A classical way around this type of intractability problem consists
of running an MCMC algorithm targeting the joint distribution $\pi(\theta,x_{1:T})$
when evaluating this density, possibly up to a constant, is feasible.
This significantly broadens the class of models under consideration
to which MCMC can be applied. There are, however, well documented
difficulties with this approach. The standard strategy consists of
updating alternately $x_{1:T}$ conditional upon $\theta$ and $\theta$
conditional upon $x_{1:T}$. As $x_{1:T}$ is a high-dimensional vector,
one typically updates it by sub-blocks using MH steps with tailored
proposal distributions \citep{shephard1997likelihood}. However, for
complex SSMs, it is very difficult to design efficient proposal distributions.
An alternative consists of using the cSMC update described in Algorithm
3 which allows one to update the state $x_{1:T}$ conditional upon
$\theta$ in one block. A strong dependence between $\theta$ and
$x_{1:T}$ may however still lead to underperforming algorithms. We
will come back to this point later in the paper.

A powerful alternative method to tackle intractability which has recently
attracted some interest consists of replacing the value of $\pi(\theta)$
with a non-negative random estimator $\hat{\pi}(\theta)$ whenever
it is required in \eqref{eq:genericMHacceptratio} for the implementation
of the marginal MH algorithm. If $\mathbb{E}[\hat{\pi}(\theta)]=C\pi(\theta)$
for all $\theta\in\Theta$ and a constant $C>0$ it turns out to lead
to exact algorithms, that is sampling from $\pi$ is guaranteed at
equilibrium under very mild assumptions on $\hat{\pi}(\theta)$. This
approach leads to so called pseudo-marginal algorithms \citep{Andrieu_and_Roberts_2009}.
As SMC provides a nonnegative unbiased estimate \eqref{eq:estimate_ell_SMC}
of $l_{\theta}(y_{1:T})$ for SSMs \citep{Del_Moral_2004}, a pseudo-marginal
approximation of the marginal MH algorithm for state-space models
is possible in this context. The resulting algorithm, the particle
marginal MH (PMMH) introduced \citet{Andrieu_Doucet_Holenstein_2009,Andrieu_et_al_2010},
is presented in Algorithm \ref{alg: PMMH} .

\begin{algorithm}
\protect\caption{PMMH for SSM}

\label{alg: PMMH} \KwIn{Current sample $(\theta,\hat{l}_{\theta}(y_{1:T}))$,
$N\geq1$ } \KwOut{New sample $(\theta',\hat{l}_{\theta'}(y_{1:T}))$}

Sample $\theta'\sim q(\theta,\cdot)$ \\
 Run $\mathrm{SMC}\big(N,\mbox{\ensuremath{\mathscr{M_{\theta'},\mathscr{A}_{\theta'}}}}\big)$
for $\pi_{\theta'}(x_{1:T})$ \\
 Compute the estimate $\hat{l}_{\theta'}(y_{1:T})$ of $l_{\theta'}(y_{1:T})$
with the output of $\mathrm{SMC}\big(N,\mbox{\ensuremath{\mathscr{M_{\theta'},\mathscr{A}_{\theta'}}}}\big)$
using \eqref{eq:estimate_ell_SMC}.

Return $(\theta',\hat{l}_{\theta'}(y_{1:T}))$ with probability 
\[
\min\left\{ 1,\frac{q(\theta',\theta)\eta(\theta')\hat{l}_{\theta'}(y_{1:T})}{q(\theta,\theta')\eta(\theta)\hat{l}_{\theta}(y_{1:T})}\right\} ,
\]
otherwise return $(\theta,\hat{l}_{\theta}(y_{1:T}))$. 
\end{algorithm}

The PMMH defines a Markov chain $\big\{\theta_{i},\hat{l}_{\theta_{i}}(y_{1:T})\big\}$
which leaves $\pi({\rm d}\theta)$ invariant marginally. However,
as shown in \citet{Andrieu_Doucet_Holenstein_2009,Andrieu_et_al_2010},
it is easy to recover samples from $\pi(\theta,x_{1:T})$ by adding
an additional step to Algorithm \ref{alg: PMMH}.

Although the PMMH has been recognised as significantly extending the
applicability of MCMC to a broader class of state-space models \citet{flury2011bayesian},
it comes with some drawbacks. In particular the performance of the
resulting MCMC algorithm depends heavily on the variability of the
induced acceptance ratio \citep{Andrieu_and_Roberts_2009,andrieu2015convergence,andrieu2014establishing,doucet2015efficient,pitt2012some,sherlock2015efficiency},
and overestimates $\hat{l}_{\theta}(y_{1:T})$ of $l_{\theta}(y_{1:T})$
lead to an algorithm rejecting many transitions away from $\theta$,
resulting in poor performance. This means for example that $N$ should
scale linearly with $T$ in order to maintain a set level of performance
as $T$ increases. In the following, we present another new class
of exact approximate MCMC algorithms targetting $\pi(\theta,x_{1:T})$,
which still update $(\theta,x_{1:T})$ jointly but can be interpreted
as using unbiased estimates of the acceptance ratio $r(\theta,\theta')$
computed afresh at each iteration of the MCMC algorithm. This lack
of memory is to be contrasted with the potentially calamitous reliance
of the PMMH's acceptance ratio on the estimate of the likelihood obtained
the last time an acceptance occurred (refreshing this quantity using
SMC would lead to an invalid algorithm, see \citet{Beaumont_2003,Andrieu_and_Roberts_2009}).
In addition, as we shall see, algorithms such as the marginal MH in
Algorithm \ref{alg: MH} requires a proposal such that the distance
between $\theta$ and $\theta'$ is of order $T^{-1/2}$ in order
to account for the concentration of the posterior distribution. This
turns out to provide us with an additional built-in beneficial mechanism
to reduce variability of our estimator of the acceptance ratio, independent
of $N$.

\subsection{AIS within Metropolis-Hastings}

In order to define a valid MH update which uses the estimators of
$\mathfrak{L}(\theta,\theta')$ described in Section \ref{sec: cSMC-AIS},
additional conditions to those of (A\ref{hypA:basicconditions}) are
required\textendash fortunately these conditions are satisfied by
the cSMC update, with or without backward sampling \citep{Chopin_and_Singh_2015}.

\begin{condition}For any $\theta,\theta'\in\Theta$, $\mathscr{P}_{\theta,\theta'}$
and $\mathscr{R}_{\theta,\theta'}$ satisfying (A\ref{hypA:basicconditions}),
and such that 
\begin{enumerate}
\item the distributions in $\mathscr{P}_{\theta,\theta'}$ satisfy $\pi_{\theta,\theta',\varsigma}(\cdot)=\pi_{\theta',\theta,1-\varsigma}(\cdot)$
for any $\varsigma\in[0,1]$ 
\item the transition kernels in $\mathscr{R}_{\theta,\theta'}$ satisfy,
for any $\varsigma\in[0,1]$,

\begin{enumerate}
\item $R_{\theta,\theta',\varsigma}(\cdot,\cdot)=R_{\theta',\theta,1-\varsigma}(\cdot,\cdot)$, 
\item $R_{\theta,\theta',\varsigma}(\cdot,\cdot)$ is $\pi_{\theta,\theta',\varsigma}-$reversible. 
\end{enumerate}
\end{enumerate}
\end{condition}

Following the setup above, the pseudocode of MCMC AIS is given in
Algorithm \ref{alg: MCMC AIS for HMM}.

\begin{algorithm}
\protect\caption{MCMC AIS for SSM}

\label{alg: MCMC AIS for HMM} \KwIn{Current sample $(\theta,x_{1:T})$,$K,\varsigma(\cdot)$}

\KwOut{New sample $(\theta',x'_{1:T})$}

Sample $\theta'\sim q(\theta,\cdot)$. \\
 $\big(x'_{1:T},\mathfrak{\hat{L}}(\theta,\theta')\big)\leftarrow{\rm AIS}\big(x_{1:T},\mathscr{P}_{\theta,\theta'},\mathscr{R}_{\theta,\theta'},K,\varsigma(\cdot)\big)$.
\\
 Return $(\theta',x'_{1:T})$ with probability $\min\{1,r_{\mathbf{u}}(\theta,\theta')\}$,
where 
\begin{equation}
r_{\mathbf{u}}(\theta,\theta')=\frac{q(\theta',\theta)}{q(\theta,\theta')}\frac{\eta(\theta')}{\eta(\theta)}\mathfrak{\hat{L}}(\theta,\theta').\label{eq: noisy acceptance ratio}
\end{equation}
Otherwise return $(\theta,x_{1:T})$. 
\end{algorithm}

It can be shown that this algorithm is reversible with respect to
$\pi(\theta,x_{1:T})$ for any $K\geq0$; see \citet{Neal_2004} and
\citet{Karagiannis_and_Andrieu_2013} for details. An important point
here is that although the approximated acceptance ratio is reminiscent
of that of a MH algorithm targeting $\pi(\theta)$, the present algorithm
targets the joint density $\pi(\theta,x_{1:T})$: the simplification
occurs only because the random variable corresponding to $\mathbf{u}_{K}$
will be approximately distributed according to $\pi_{\theta'}(\cdot)$
when $K$ is large enough, under proper mixing conditions. When $K=0$
this transition leads to a reducible algorithm since $x_{1:T}$ is
not updated. However this scheme can be used as part of a Metropolis-within-Gibbs
where $x_{1:T}$ is updated conditional upon the parameter using,
say, $R_{\theta}(\cdot,\cdot)$. We will refer to the latter algorithm
for which $R_{\theta}$ is a cSMC with backward sampling as Metropolis-within-Particle-Gibbs
(MwPG) in the rest of the paper. 
\begin{rem}
In the scenario where a cSMC procedure involving $N$ particles is
used, the algorithm above may seem wasteful as only one particle is
used in order to approximate the likelihood ratio $\mathfrak{L}(\theta,\theta')$
in \eqref{eq: noisy acceptance ratio}. Ideally one would want to
use $M>1$ particles and average $M$ likelihood ratio estimators
in order to reduce variability and improve the properties of the algorithm.
Using this averaged estimator of the likelihood ratio in Algorithm
\ref{alg: MCMC AIS for HMM} would, however, lead to a Markov kernel
which does preserve $\pi(\theta,x_{1:T})$ as an invariant density.
A novel methodology allowing the use of such averaged estimators within
MCMC has been developed in \citet{andrieu2016alternativepseudo}. 
\end{rem}

\section{A theoretical analysis}

In this section we develop an analysis of the likelihood ratio estimator
and of the MCMC AIS algorithm in a scenario which can be treated rigorously
in a few pages, but yet is of practical interest\textendash in particular
our findings are supported empirically by the simulations of Section
\ref{sec: Numerical examples}, where more general scenarios are considered,
and shed some light on some of our empirical results. Extension to
more general scenarios is however far beyond the scope of the present
manuscript. We consider the scenario where for any $\theta\in\Theta$,
$f_{\theta}(x_{t-1},x_{t})$ is independent of $x_{t-1}$, that is
for any $T\geq1$ 
\[
p_{\theta}\big(x_{1:T},y_{1:T}\big)={\displaystyle \prod\limits _{t=1}^{T}p_{\theta}\big(x_{t},y_{t}}\big),
\]
with 
\[
p_{\theta}\big(x_{t},y_{t}\big):=f_{\theta}\big(x_{t}\big)g_{\theta}\big(y_{t}\mid x_{t}\big).
\]
We define the conditional distributions $\{\pi_{\theta,T}(x_{1:T};\omega)\propto p_{\theta}\big(x_{1:T},y_{1:T}\big),T\geq1\}$
where $\omega:=\left\{ y_{t},t\geq1\right\} \subset\mathsf{Y}^{\mathbb{N}}$.
We further assume that the marginal MH algorithm underpinning our
update is a random walk Metropolis (RWM) algorithm and that $K=1$.
Our aim is to show that as $T\rightarrow\infty$ the algorithm does
not degenerate, in a sense to be made more precise below, provided
the RWM proposal distribution is properly scaled with $T$ and $N_{T}$
sufficiently large, where $N_{T}$ is the number of particles used
in the cSMC. In particular $N_{T}$ is not required to grow with $T$,
as observed in simulations\textendash see Theorem \ref{thm:main}
for a precise formulation of our result. This should be contrasted
with results from the simulated likelihood literature where the condition
$\sqrt{T}/N_{T}=o\left(1\right)$ is necessary to ensure asymptotic
efficiency of the maximum simulated likelihood estimator \citep{flury2011bayesian,lee-lf-1992}
. We now introduce some notation useful in order to formulate and
prove our result. The intermediate distribution is defined as 
\[
\gamma_{\theta,\theta',1}\big(x_{1:T}\big):=p_{(\theta+\theta')/2}\big(x_{1:T},y_{1:T}\big)\,;
\]
it will be clear from our proof that this is in no way a restriction
but has the advantage of keeping our development as simple as possible.
To define our RWM we require an increment proposal distribution based
on a symmetric increment distribution $q_{0}(\cdot)$ (independent
of $T$) and such that $q_{T}(\theta,\theta'):=\sqrt{T}q_{0}\big(\sqrt{T}(\theta-\theta')\big)$.
It will be convenient in what follows to define a proposed sample
in the following way: for any $(\theta,\epsilon)\in\Theta\times\Xi$
($\epsilon$ will be distributed according to $q_{0}(\cdot)$) we
let 
\[
\theta^{\prime}(\epsilon,T):=\theta+\frac{\epsilon}{\sqrt{T}}\;\text{and}\;\mbox{\ensuremath{\tilde{\theta}}}(\epsilon,T):=\frac{\theta+\theta'(\epsilon,T)}{2}.
\]
For simplicity of presentation we assume that $\inf_{(\theta,x,y)\in\Theta\times\mathsf{X}\times\mathsf{Y}}\eta(\theta)p_{\theta}\big(x,y\big)>0$.
As a result for any $(\theta,\epsilon)\in\Theta\times\Xi$ and $\omega\in\mathsf{Y}^{\mathbb{N}}$
we let 
\[
r_{T}(\theta,\epsilon;\omega):=\frac{\eta(\theta'(\epsilon,T))p_{\theta^{\prime}(\epsilon,T)}\big(y_{1:T}\big)}{\eta(\theta)p_{\theta}\big(y_{1:T}\big)}
\]
be the marginal acceptance ratio, which is zero whenever $\theta'(\epsilon,T)\notin\Theta$.
For $\xi:=\left\{ (x_{t},x_{t}'),t\geq1\right\} \subset\big(\mathsf{X}\times\mathsf{X}\big)^{\mathbb{N}}$
the acceptance ratio of the MCMC-AIS algorithm can be written as 
\[
\tilde{r}_{T}(\theta,\epsilon;\omega,\xi):=r_{T}(\theta,\epsilon;\omega)\exp\big(\Lambda_{T}\big(\theta,\epsilon;\omega,\xi\big)\big)
\]
where for $\theta,\theta'(\epsilon,T)\in\Theta$, 
\begin{align}
\Lambda_{T}\big(\theta,\epsilon;\omega,\xi\big) & :=\log\frac{p_{\tilde{\theta}(\epsilon,T)}\big(x_{1:T}\mid y_{1:T}\big)}{p_{\theta}\big(x_{1:T}\mid y_{1:T}\big)}+\log\frac{p_{\theta^{\prime}(\epsilon,T)}\big(x_{1:T}^{\prime}\mid y_{1:T}\big)}{p_{\tilde{\theta}(\epsilon,T)}(x_{1:T}^{\prime}\mid y_{1:T})}\label{eq:defLambda}\\
 & =\sum_{t=1}^{T}\left\{ \log\frac{p_{\tilde{\theta}(\epsilon,T)}\big(x_{t}\mid y_{t}\big)}{p_{\theta}\big(x_{t}\mid y_{t}\big)}+\log\frac{p_{\theta^{\prime}(\epsilon,T)}\big(x_{t}^{\prime}\mid y_{t}\big)}{p_{\tilde{\theta}(\epsilon,T)}\big(x_{t}^{\prime}\mid y_{t}\big)}\right\} .\nonumber 
\end{align}
In order to limit the amount of notation we will not distinguish between
random variables and their realisations using small/capital letters
whenever Greek letters are used. For any $(\theta,y)\in\Theta\times\mathsf{Y}$
and $N\geq1$ we let $R_{\theta,y}^{[N]}:\mathsf{X}\times\mathcal{X}\rightarrow[0,1]$
denote an MCMC kernel targeting the probability distribution of density
$p_{\theta}\big(\cdot\mid y\big)$ using a tuning parameter $N$ governing
its ergodicity properties: we have here in mind a conditional SMC
using $N$ particles, but this will not be a requirement (one could
iterate a given ergodic and reversible kernel $N$ times for example).
Now for any $\omega\in\mathsf{Y}^{\mathbb{N}}$ and $T\geq1$ we define
the process $\xi_{T}:=\left\{ (X_{t},X_{t}'),t\geq1\right\} $ as
a sequence of independent random vectors with marginal laws given
by $\mathbb{P}_{\theta,\epsilon,T}^{\omega}\big((X_{t},X_{t'})\in A\big):=\int_{A}p_{\theta}\big({\rm d}x\mid y_{t}\big)R_{\tilde{\theta}(\epsilon,T),y_{t}}^{[N_{T}]}\big(x,{\rm d}x'\big)$\textendash we
omit the dependence of $(X_{t},X_{t}')$ on $T$ (and $\epsilon$)
for notational simplicity, may write $\xi$ for $\xi_{T}$ when no
ambiguity is possible, but we should bear in mind that we will deal
with triangular arrays of random variables in what follows. We let
$\mathbb{P}_{\theta,\epsilon,T}^{\omega}(\cdot),\mathbb{E}_{\theta,\epsilon,T}^{\omega}(\cdot)$,
$\mathbb{C}_{\theta,\epsilon,T}^{\omega}(\cdot,\cdot)$ and $\mathbb{V}_{\theta,\epsilon,T}^{\omega}(\cdot)$
be the probability, expectation covariance and variance of the process
$\xi$ conditional upon a realisation of $\omega\in\mathsf{Y}^{\mathbb{N}}$\textendash we
may drop $\epsilon,T$ when unnecessary e.g. when considering events
involving $\left\{ X_{t},t\geq1\right\} $ only. Further we consider
$\big\{ Y_{t},t\geq1\big\}$ a sequence of independent and identically
distributed random variables taking their values in $\mathsf{Y}$
(and $\sigma-$algebra $\mathcal{Y}$) and we denote the corresponding
probability distribution $P$. Let $\mathcal{N}\left(\mu,\Sigma\right)$
denote the normal distribution of mean $\mu$ and covariance $\Sigma$.
In essence we show that $P-$a.s., for any suitable $(\theta,\epsilon)\in\Theta\times\Xi$
and an independent sequence $\big\{\xi_{\tau},\tau\in\mathbb{N}\big\}$
where $\xi_{\tau}\sim\mathbb{P}_{\theta,\epsilon,\tau}^{\omega}$
we have that the law of $\Lambda_{T}\big(\theta,\epsilon;\omega,\xi_{T}\big)$
can be approximated to arbitrary precision by $\mathcal{N}\big(-\sigma^{2}(\theta,\epsilon)/2,\sigma^{2}(\theta,\epsilon)\big)$
(for some constant $\sigma^{2}(\theta,\epsilon)<\infty$ independent
of $\omega$) for $T\geq T_{0}$ and $N_{T}\geq N_{0}$ where $N_{0},T_{0}\in\mathbb{N}$
are sufficiently large. In particular $N_{T}$ is not required to
grow with $T$. This suggests that at equilibrium and for sufficiently
large $T$ and $N$ our algorithm behaves similarly to the penalty
method \citep{ceperley1999penalty} with acceptance probability 
\begin{equation}
\min\big\{1,r_{T}(\theta,\epsilon;\omega)\exp\big(Z\big)\big\}\label{eq:penalty-like-acceptance}
\end{equation}
with $Z\mid(\theta,\epsilon,\omega)\sim\mathcal{N}\left(-\varsigma_{T}^{2}(\theta,\epsilon)/2,\varsigma_{T}^{2}(\theta,\epsilon)\right)$
for some sequence $\varsigma_{T}^{2}(\theta,\epsilon)\rightarrow\sigma^{2}(\theta,\epsilon)$
as $T$ increases, although in our scenario the Markov chain considered
consists of both the parameter $\theta$ and the states $x_{1:T}$,
not just the parameter as for the method presented in \citet{delagiannidis:doucet:pitt:2015}.
As a result, if the marginal algorithm scales with $T$ we see that
our algorithm also scales, and only incurs a penalty independent of
$T$. This is the case under the general conditions of \citet[Lemma 19.31]{vdV00}
and ideas of \citet[Lemma 2.1]{kleijn2012} as a local asymptotic
normality in the misspecified scenario can be applied and leads to
the expansion, with $\dot{\ell}_{\theta}(y):=\partial_{\theta}\log p_{\theta}\left(y\right)$,
$\Theta\subset\mathbb{R}$ and some constant $V(\theta)>0$

\[
\log\frac{l_{\theta^{\prime}(\epsilon,T)}\left(Y_{1:T}\right)}{l_{\theta}\left(Y_{1:T}\right)}=\frac{\epsilon}{\sqrt{T}}\sum_{i=1}^{T}\dot{\ell}_{\theta}\big(Y_{i}\big)-\frac{1}{2}\epsilon^{2}V(\theta)+o_{P}(1),
\]
which together with a continuity assumptions on the prior density
$\eta(\theta)$ suggests again a central limit theorem, and hence
the fact that the acceptance ratio converges to a log-normal random
variable independent of $T$ . We do not focus on this latter problem,
but establish that our algorithms behaves similarly to the algorithm
with acceptance ratio given in \eqref{eq:penalty-like-acceptance}
as $T$ and $N_{T}$ are sufficiently large, both in terms of expected
acceptance probability and relative mean square jump distance (or
equivalently first order autocorrelation)\textendash see Theorem \ref{thm:main}.

We let $\ell_{\theta}\big(x\mid y\big):=\log p_{\theta}(x\mid y)$,
$\dot{\ell}_{\theta}(x\mid y):=\partial_{\theta}\log p_{\theta}\left(x\mid y\right)$,
$\ddot{\ell}_{\theta}(x\mid y):=\partial_{\theta}^{2}\log p_{\theta}\left(x\mid y\right)$,
$\dddot{\ell}_{\theta}(x\mid y):=\partial_{\theta}^{3}\log p_{\theta}\left(x\mid y\right)$,
and similarly $\ell_{\theta}\big(y\big):=\log p_{\theta}(y)$, $\dot{\ell}_{\theta}(y):=\partial_{\theta}\log p_{\theta}\left(y\right)$
and $\ddot{\ell}_{\theta}(y):=\partial_{\theta}^{2}\log p_{\theta}\left(y\right)$.
The total variation distance is defined for any probability distributions
$\nu_{1},\nu_{2}$ on $(\mathsf{X},\mathcal{X})$ as $\|\nu_{1}-\nu_{2}\|_{tv}:=\frac{1}{2}\sup_{f:\mathsf{X}\rightarrow[-1,1]}[\nu_{1}(f)-\nu_{2}(f)]$.
We require the following assumptions for our analysis.

\begin{condition}\label{hyp:strong-likelihood-cond} 
\begin{enumerate}
\item \label{enu:-ThetaConvex}$\Theta\subset\mathbb{R}$ and $\Xi\subset\mathbb{R}$
are compact sets, $\Theta$ is convex, $\mathsf{X}\subset\mathbb{R}^{d_{x}}$
and $\mathsf{Y}\subset\mathbb{R}^{d_{y}}$ for some $d_{x},d_{y}\in\mathbb{N}$. 
\item $q_{0}(\cdot)$ is a symmetric probability distribution, bounded away
from zero. 
\item \label{enu:regularityell}$\inf_{(\theta,x,y)\in\Theta\times\mathsf{X}\times\mathsf{Y}}p_{\theta}\big(x,y\big)>0$
and for any $x,y\in\mathsf{X}\times\mathsf{Y}$, $\theta\mapsto\ell_{\theta}(x,y)$
is three times differentiable with 
\[
\bar{\ell}^{(1)}:=\underset{\left(\theta,x,y\right)\in\Theta\times\mathsf{X}\times\mathsf{Y}}{\sup}\big|\dot{\ell}_{\theta}\big(x\mid y\big)\big|<\infty,\quad\bar{\ell}^{(2)}:=\underset{\left(\theta,x,y\right)\in\Theta\times\mathsf{X}\times\mathsf{Y}}{\sup}\big|\ddot{\ell}_{\theta}\big(x\mid y\big)\big|<\infty
\]
and 
\[
\bar{\ell}^{(3)}:=\underset{\left(\theta,x,y\right)\in\Theta\times\mathsf{X}\times\mathsf{Y}}{\sup}\left\vert \dddot{\ell}_{\theta}\big(x\mid y\big)\right\vert <\infty,
\]
\item $\theta,x,y\mapsto\dot{\ell}_{\theta}\big(x\mid y\big),\ddot{\ell}_{\theta}\left(x\mid y\right)$
and $\dddot{\ell}_{\theta}\big(x\mid y\big)$ are measurable, 
\item \label{enu:variancepositiveforallthetas}for all $\theta\in\Theta$
and $\omega\in\mathsf{Y}^{\mathbb{N}}$, $\mathbb{E}_{\theta}^{\omega}\left[\dot{\ell}_{\theta}(X_{1}\mid y_{1})\right]=0$,
$\mathbb{V}_{\theta}^{\omega}\left[\dot{\ell}_{\theta}(X_{1}\mid y_{1})\right]=-\mathbb{E}_{\theta}^{\omega}\left[\ddot{\ell}_{\theta}(X_{1}\mid y_{1})\right]$
and $\inf_{(\theta,y_{1})\in\Theta\times\mathsf{Y}}\mathbb{V}_{\theta}^{\omega}\left[\dot{\ell}_{\theta}(X_{1}\mid y_{1})\right]>0$, 
\item \label{enu:propertieskernel}$R_{\theta,y}^{[N]}$ is a $p_{\theta}(\cdot\mid y)-$reversible
Markov transition probability and 
\[
\lim_{N\rightarrow\infty}\sup_{(\theta,x,y)\in\Theta\times\mathsf{X}\times\mathsf{Y}}\|R_{\theta,y}^{[N]}\big(x,\cdot\big)-p_{\theta}(\cdot\mid y)\|_{tv}=0.
\]
\end{enumerate}
\end{condition}

Some of these conditions are restrictive in the sense that the required
uniformity in $\theta,\omega,\xi$, exploited here to keep the proof
short, implicitly imposes boundedness of these variables; we discuss
this in more detail in subsection \ref{subsec:DiscussionAssumptions}
and explain how our results can be extended to more general scenarios
without changing our proof strategy and the nature of the result,
but at the expense of significant additional technical complications.

For $\omega\in\mathsf{Y}^{\mathbb{N}}$ we let $\mathbb{E}_{T}^{\omega}(\cdot)$
be the expectation such that for any measurable function $f:\Theta\times\Xi\times\mathsf{X}^{\mathbb{N}}\rightarrow\mathbb{R}$
\[
\mathbb{E}_{T}^{\omega}\big[f(\theta,\epsilon,\xi)\big]=\int\mathbb{E}_{\theta,\epsilon,T}^{\omega}\big[f(\theta,\epsilon,\xi)\big]q_{0}({\rm d}\epsilon)\pi_{_{T}}({\rm d}(\theta,x_{1:T});\omega)R_{\tilde{\theta}(\epsilon,T),\omega,T}^{[N_{T}]}(x_{1:T},{\rm d}x'_{1:T})
\]
where $R_{\theta,\omega,T}^{[N]}(x_{1:T},\cdot):=\prod_{t=1}^{T}R_{\theta,y_{t}}^{[N]}\big(x_{t},\cdot\big)$.
Finally for $f:\Theta\times\Xi\times\mathsf{X}^{\mathbb{N}}\times\mathsf{Y}^{\mathbb{N}}\rightarrow\mathbb{R}$
we define 
\[
\mathbb{E}_{T}\big[f(\theta,\epsilon,\xi,\omega)\big]:=\int\mathbb{E}_{T}^{\omega}\big[f(\theta,\epsilon,\xi,\omega)\big]P({\rm d}\omega)
\]
and for $f:\Theta\times\mathsf{X}\times\mathsf{Y}\rightarrow\mathbb{R}$
\[
\mathbb{E}_{\theta}\big[f(\theta,X_{1},Y_{1})\big]:=\int\mathbb{E}_{\theta,\epsilon,T}^{\omega}\big[f(\theta,X_{1},Y_{1})\big]P({\rm d}\omega).
\]
We establish the following result. 
\begin{thm}
\label{thm:main}Assume (A\ref{hyp:strong-likelihood-cond}). Then
$P-$a.s., for any $\varepsilon_{0}>0$ there exist $T_{0},N_{0}\in\mathbb{N}$
such that for any $T\geq T_{0}$ and any sequence $\big\{ N_{T}\big\}\in\mathbb{N}^{\mathbb{N}}$
such that $N_{T}\geq N_{0}$ for $T\geq T_{0}$ 
\[
\sup_{T\geq T_{0}}\left|\mathbb{E}_{T}^{\omega}\left[\min\{1,\tilde{r}_{T}(\theta,\epsilon;\omega,\xi)\}\right]-\check{\mathbb{E}}_{T}^{\omega}\left[\min\{1,r_{T}(\theta,\epsilon;\omega)\exp(Z)\}\right]\right|\leq\varepsilon_{0},
\]
and 
\[
\sup_{T\geq T_{0}}\left|\mathbb{E}_{T}^{\omega}\left[\min\{1,\tilde{r}_{T}(\theta,\epsilon;\omega,\xi)\}\epsilon^{2}\right]-\check{\mathbb{E}}_{T}^{\omega}\left[\min\{1,r_{T}(\theta,\epsilon;\omega)\exp(Z)\}\epsilon^{2}\right]\right|\leq\varepsilon_{0}
\]
where $\check{\mathbb{E}}_{T}^{\omega}[f(\theta,\epsilon,Z)]:=\mathbb{E}_{T}^{\omega}\big[\check{\mathbb{E}}_{\theta,\epsilon}^{\omega}[f(\theta,\epsilon,Z)]\big]$
with, for $(\theta,\epsilon,\omega)\in\Theta\times\Xi\times\mathsf{Y}^{\mathbb{N}}$,
$\check{\mathbb{E}}_{\theta,\epsilon}^{\omega}[\cdot]$ the conditional
expectation of 
\[
Z\mid(\theta,\epsilon,\omega)\sim\mathcal{N}\left(-\frac{\varsigma_{T}^{2}(\theta,\epsilon)}{2},\varsigma_{T}^{2}(\theta,\epsilon)\right)
\]
where $\varsigma_{T}^{2}(\theta,\epsilon):=\sigma^{2}(\tilde{\theta}(\epsilon,T),\epsilon)$
with 
\[
\sigma^{2}(\theta,\epsilon):=\frac{-\epsilon^{2}}{2}\mathbb{E}_{\theta}\left[\ddot{\ell}_{\theta}(X_{1}\mid Y_{1})\right].
\]
\end{thm}
\begin{rem}
We remark that the (renormalized) expected mean square jump distance
is typically asymptotically proportional to the second quantity considered
above, since 
\[
\frac{\mathbb{E}_{T}^{\omega}\big[\min\{1,\tilde{r}_{T}(\theta,\epsilon;\omega,\xi)\}\,(\theta'(\epsilon,T)-\theta)^{2}\big]}{\mathbb{V}_{T}^{\omega}(\theta)}=\frac{\mathbb{E}_{T}^{\omega}\big[\min\{1,\tilde{r}_{T}(\theta,\epsilon;\omega,\xi)\}\,\epsilon^{2}\big]}{T\mathbb{V}_{T}^{\omega}(\theta)}
\]
and the fact that under standard regularity conditions we expect the
last denominator to converge to a constant. 
\end{rem}

\begin{rem}
One expects the MCMC AIS algorithm to suffer less from the dependence
between the parameter and the latent variables than the MwPG version.
However there is another advantage, observed empirically in simulations,
which can be explained theoretically in the light of our simple analysis.
One notices that in the MwPG scenario, analysis of the acceptance
ratio at equilibrium involves a term similar to the first term in
the expression for $\Lambda_{T}(\theta,\epsilon;\omega,\xi)$ in \eqref{eq:defLambda},
but where $\tilde{\theta}(\epsilon,T)$ is now replaced with $\theta'(\epsilon,T)$.
As a result, for $\theta,\epsilon\in\Theta\times\Xi$, by revisiting
our proof of Theorem \ref{thm:main}, the asymptotic distribution
of the approximating algorithm can be found to be $\mathcal{N}\big(-\sigma^{2}(\theta,\epsilon),2\sigma^{2}(\theta,\epsilon)\big)$
instead of $\mathcal{N}\big(-\sigma^{2}(\theta,\epsilon)/2,\sigma^{2}(\theta,\epsilon)\big)$
since the attempted jump is not of size $\epsilon/(2\sqrt{T})$, but
$\epsilon/\sqrt{T}$. We note that this result does not require $N_{T}$
to have a minimum value, in contrast with the result of Theorem \ref{thm:main},
but it should be clear that the choice of $N_{T}$ will affect the
performance of the algorithm. The MCMC-AIS method requires $N_{T}$
to be sufficiently large in order to ensure that the dependence between
the first and second term involved in \eqref{eq:defLambda} is sufficiently
small. 
\end{rem}

\section{Numerical examples\label{sec: Numerical examples} }

In subsection \ref{sec: Experiments on an i.i.d. model} we illustrate
our theoretical findings on a simple model which in addition lends
itself to a direct comparison of MwPG and MCMC AIS, which correspond
respectively to $K=0$ and $K>0$, and allows us in particular to
assess the effect of the posterior dependence structure on $\theta$
and $x_{1:T}$ on the performance of the algorithm. In subsection
\ref{sec: Experiments on a non-linear state-space model} we compare
the algorithms proposed on a non-linear state-space model and assess
the scalability of the algorithms in terms of the number of data points
$T$.

\subsection{Experiments on an i.i.d.\ model\label{sec: Experiments on an i.i.d. model}}

Let $\mathcal{N}(z;\mu,\sigma^{2})$ denote the probability density
of a normal distribution of mean $\mu$, variance $\sigma^{2}$ and
argument $z$. We consider the simple model for which $f_{\theta}(x_{t-1},x_{t})=f_{\theta}(x_{t})=\mathcal{N}(x_{t};(1-a)\theta,\sigma_{x}^{2})$,
$\mu_{\theta}(x_{1})=f_{\theta}(x_{1})$, $g_{\theta}(x_{t},y_{t})=\mathcal{N}(y_{t};a\theta+x_{t},\sigma_{y}^{2})$
and $\eta(\theta)=\mathcal{N}(\theta;\mu_{\theta},\sigma_{\theta}^{2})$
where $a\in[0,1]$. The marginal posterior distribution $\pi(\theta)$
is invariant to the choice of $a$, but the choice of $a$ is known
to have important consequences on the posterior dependence of $\theta$
and $x_{1:T}$ \citep{gelfand1995efficient}, and hence the mixing
properties of the Gibbs sampler, that is an MCMC algorithm which alternates
sampling from $\pi(\theta\mid x_{1:T})$ and $\pi(x_{1:T}\mid\theta)$.
Indeed, as shown in \citet{Papaspiliopoulos_et_al_2003}, when $\sigma_{y}^{2}/\sigma_{x}^{2}$
is very large the choice $a\approx1$ is best while when $\sigma_{y}^{2}/\sigma_{x}^{2}$
is small the choice $a\approx0$ is preferable. For the experiments
in this section, we generated artificial data using $\sigma_{y}^{2}=0.01$
and $\sigma_{x}^{2}=1$, making $a\approx0$ optimal. We first compared
MCMC AIS cSMC-BS with $K=1$ and MwPG, whose computational complexities
per iteration are comparable provided that the cost of calculating
the acceptance ratio is much less than that of an iteration of the
cSMC-BS. For MCMC AIS cSMC-BS, the intermediate distribution is chosen
to be $\gamma_{\theta,\theta',1}=\gamma_{(\theta+\theta')/2}$ for
all $\theta,\theta'\in\Theta$. The prior variance was chosen to be
$\sigma_{\theta}^{2}=10^{5}$, therefore leading to a posterior variance
for $\theta$, $1/(1/\sigma_{\theta}^{2}+T/(\sigma_{x}^{2}+\sigma_{y}^{2}))\approx1/T$
as long as $\sigma_{x}^{2}+\sigma_{y}^{2}$ is close to $1$, the
proposal variance of the RWM is the variance of the posterior and
the particles in the cSMC routine were sampled from the prior distribution
for $x_{t}$ conditional on $\theta$, that is $M_{\theta}(x_{t-1},x_{t})=\mathcal{N}(x_{t};(1-a)\theta,\sigma_{x}^{2})$.
We first considered the scenario $a=1$, which is expected to be unfavourable
to the MwPG algorithm, and ran both algorithms once for $10^{5}$
iterations and a fine grid of values for $(N,T)$, $T=1,10,100,1000$
and $N=1,\ldots,500$. Estimates of the integrated autocorrelation
(IAC) times and expected acceptance probabilities for all scenarios
are reported in Figure \ref{fig: AIS vs MwPG PNCNHM informative observations N vs T single run}.
Despite the noisy results, a consequence of us considering only one
MCMC run per $(N,T)$ value, one can make the following observations.
As predicted by our theory, both algorithms seem to be largely insensitive
to $T$ for sufficiently large values of $N$, and while MwPG seems
to reach its asymptotic regime for smaller values of $N$, and beat
MCMC AIS cSMC for such values, MCMC AIS cSMC is more responsive to
an increase in $N$ and very rapidly beats MwPG, although not in an
apparently spectacular way.

\begin{figure}
\centerline{\includegraphics[scale=0.6]{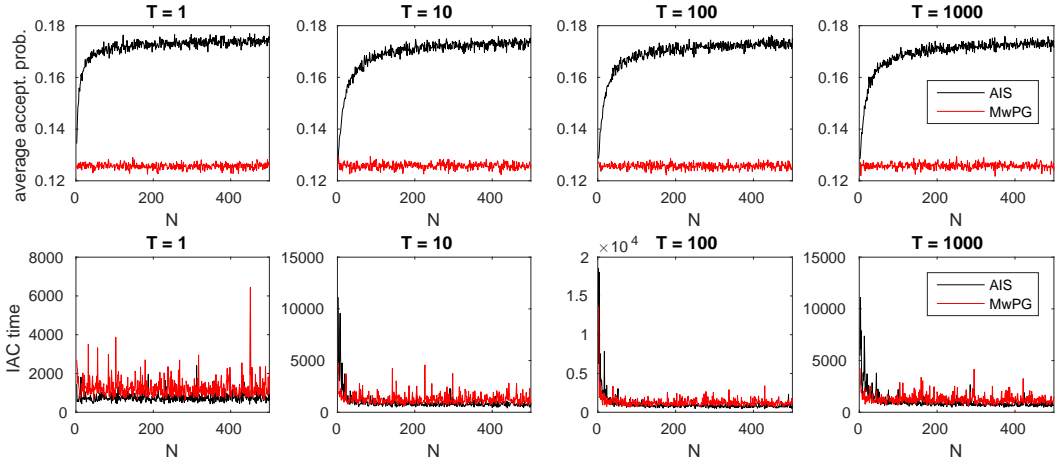}}
\protect\caption{Average rate of acceptance and IAC time for the non-centred parametrisation
of the model with informative observations.}

\label{fig: AIS vs MwPG PNCNHM informative observations N vs T single run} 
\end{figure}

We re-ran these experiments on a coarser grid of values of $(N,T)$,
more precisely all the combinations of $T=10,100,1000,10000$ and
$N=2,25,50,100,200$, but considered this time $200$ runs of the
algorithm for each such combination. The results are reported in Figure
\ref{fig: AIS vs MwPG PNCNHM informative observations N vs T multiple runs ratios included}
where we now also report in addition the ratios (MCMC AIS/MwPG) of
the mean IAC times and mean square jump distances (multiplied by $T$).
We see that the MCMC AIS algorithm is uniformly better in terms of
MSJD, while MwPG seems to be superior for small values of $N$, but
remind reader of the difficulty inherent to the estimation of IAC
and note the presence of a significant number of outliers which indicate
to us that the chains are not mixing well for such a range of values
of $N$. The algorithms' acceptance rates, not shown here, follow
a very similar pattern to that observed for the mean square jumps.

\begin{figure}
\centerline{\includegraphics[scale=0.5]{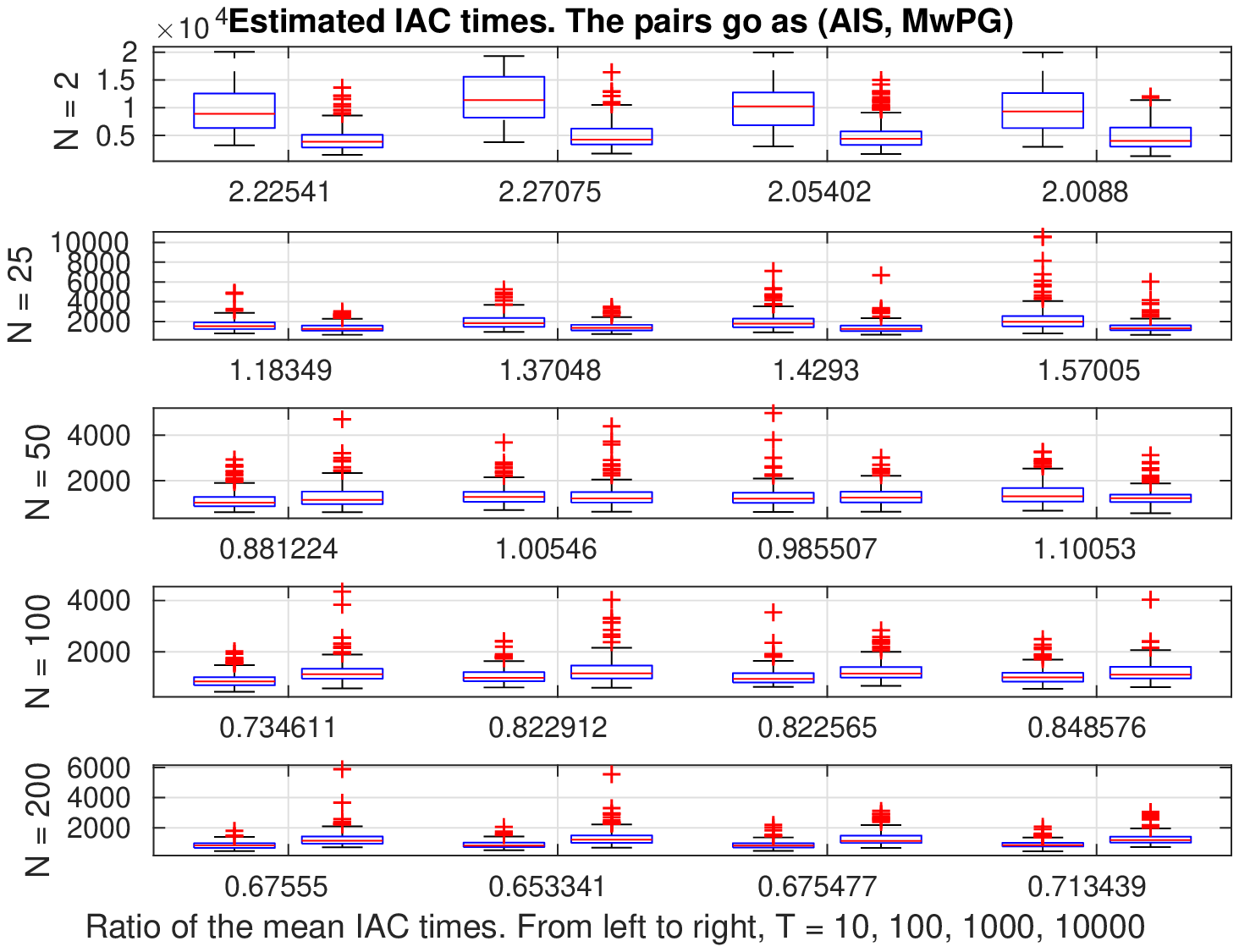}
\hspace{-1cm} \includegraphics[scale=0.5]{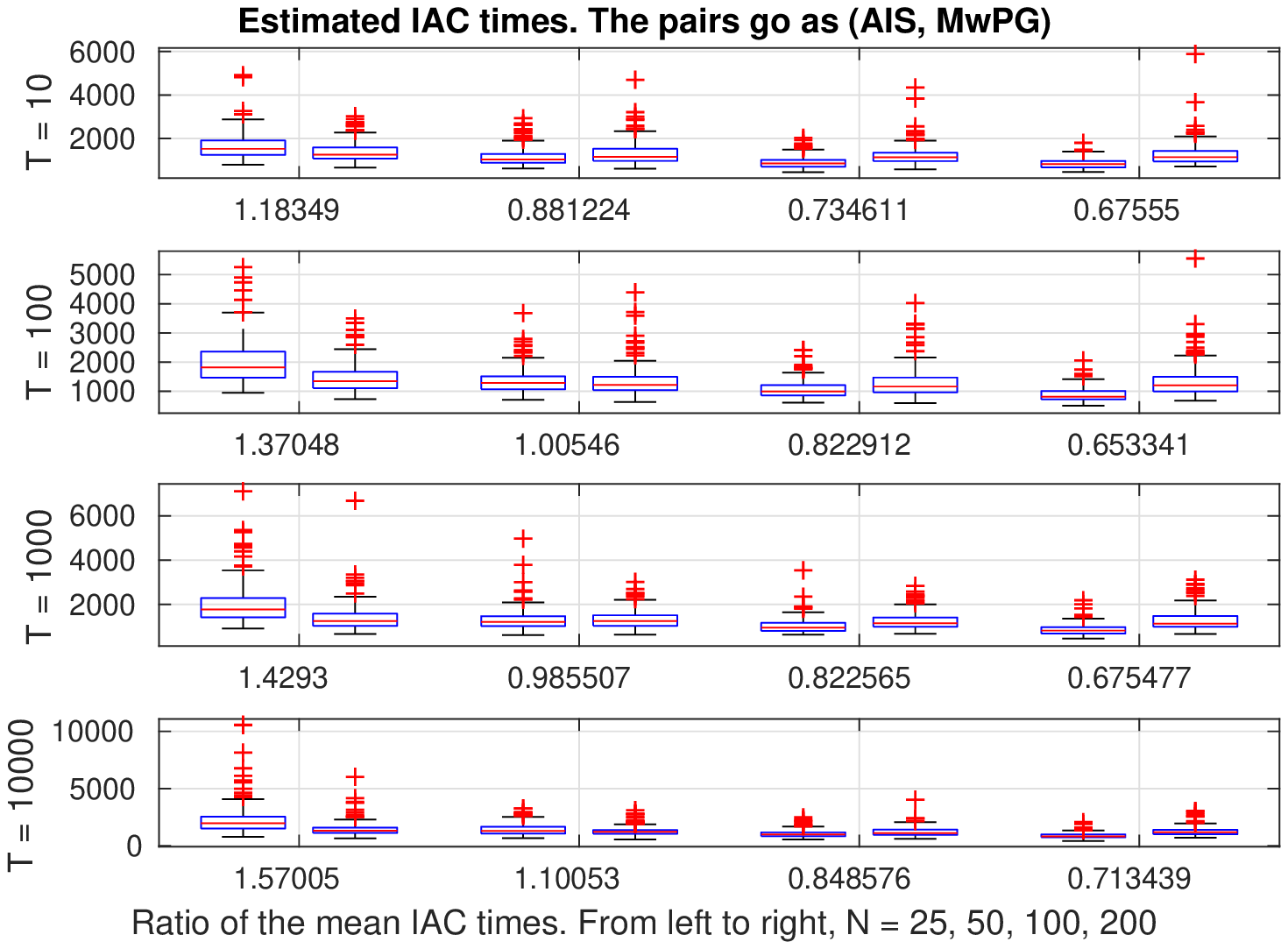}}
\centerline{\includegraphics[scale=0.5]{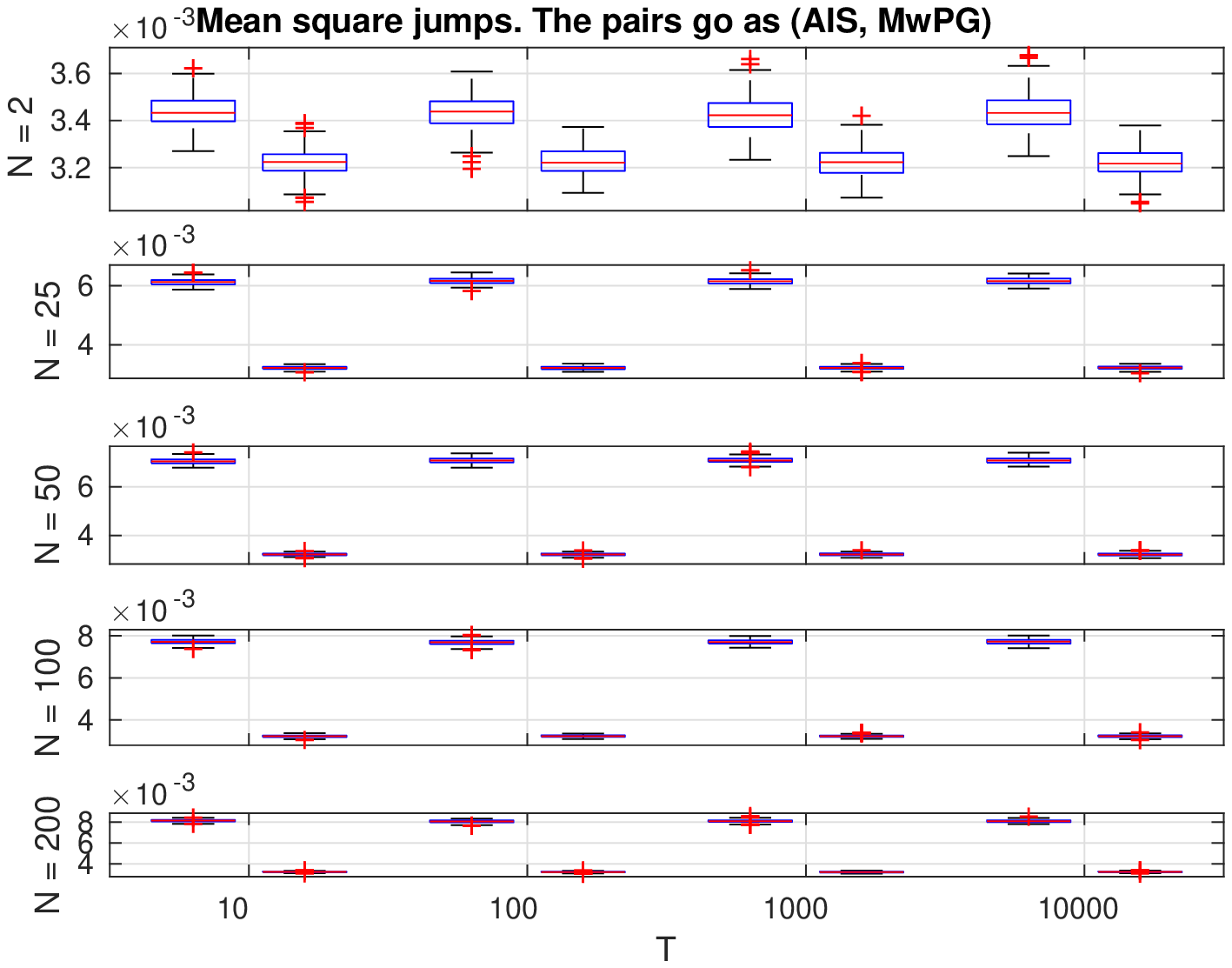}
\hspace{-1cm} \includegraphics[scale=0.5]{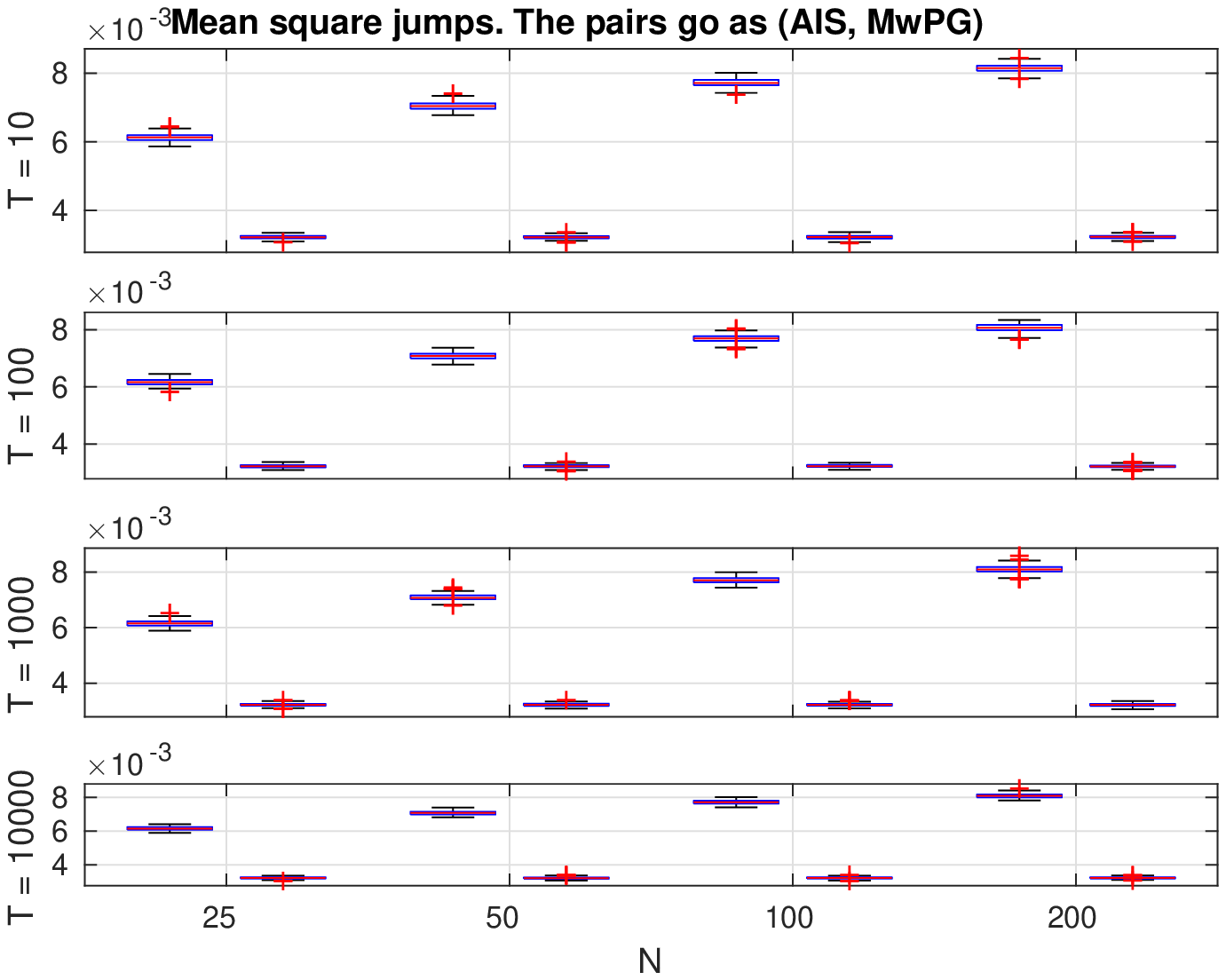}}
\protect\caption{IAC times and mean square jumps (multiplied by $T$) for the hierarchical
model for $a=1$. Ratios of the mean IAC times (MCMC AIS cSMC-BS/MwPG)
are shown on the $x-$axis. On the right hand side plots, results
for $N=2$ are not shown to improve legibility).}

\label{fig: AIS vs MwPG PNCNHM informative observations N vs T multiple runs ratios included} 
\end{figure}

We re-ran these experiments for $a=0.1$, which is more favourable
to the MwPG as this reduces the posterior dependence between $\theta$
and $x_{1:T}$. The results are presented in Figure \ref{fig: AIS vs MwPG PNCNHM informative observations N vs T multiple runs ratios included2}.
We observe that while MCMC AIS remains uniformly superior in terms
of mean square jump distance (MSJD), as expected, the IAC ratios are
now closer to one for large values of $N$, confirming that the wider
gaps observed in our earlier experiments are attributable to the posterior
dependence. This leads us to conclude that MCMC AIS is a more reliable
method than MwPG when this dependence is a priori unknown.

\begin{figure}
\centerline{\includegraphics[scale=0.5]{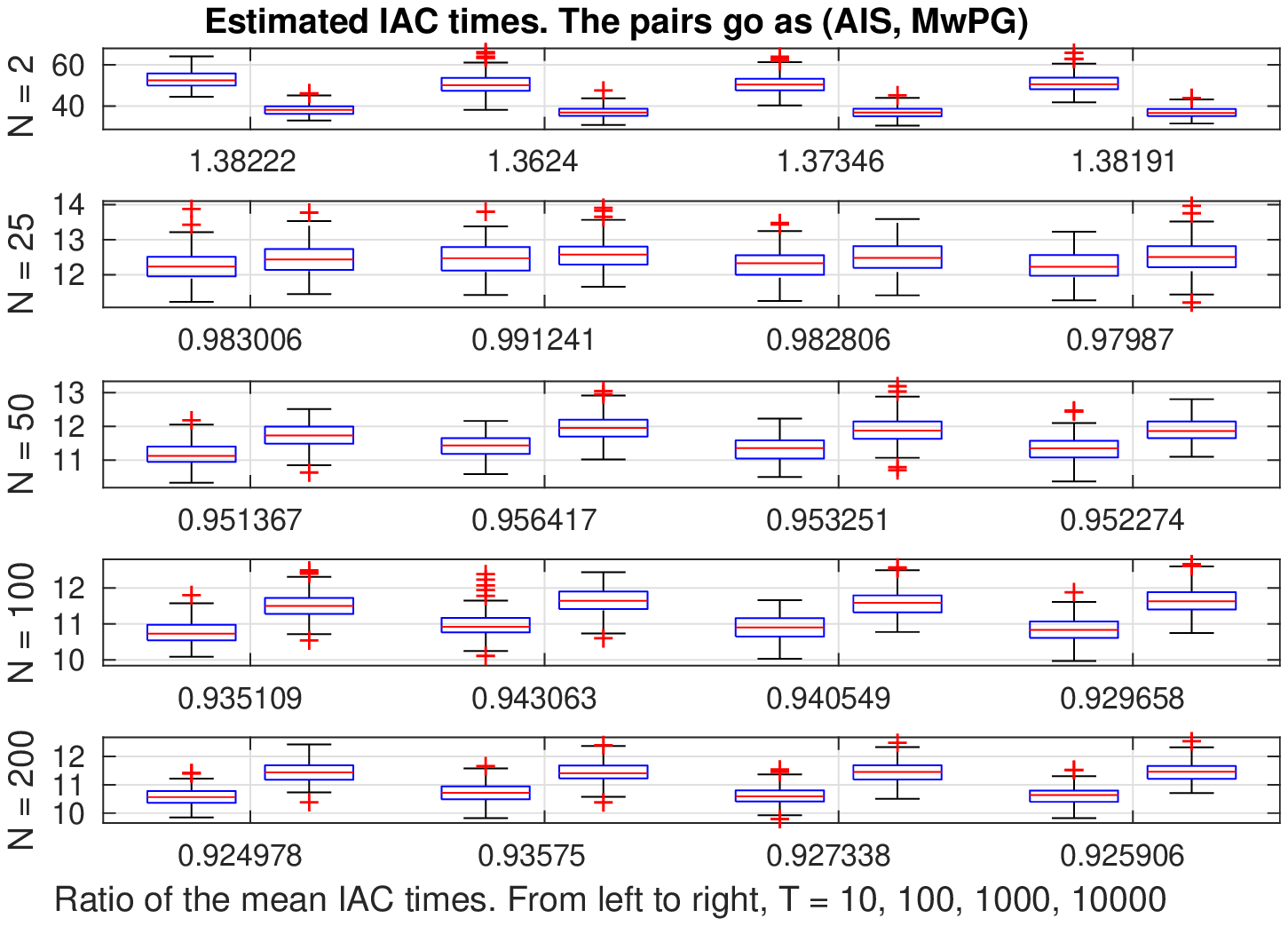}
\hspace{-0.8cm}\includegraphics[scale=0.5]{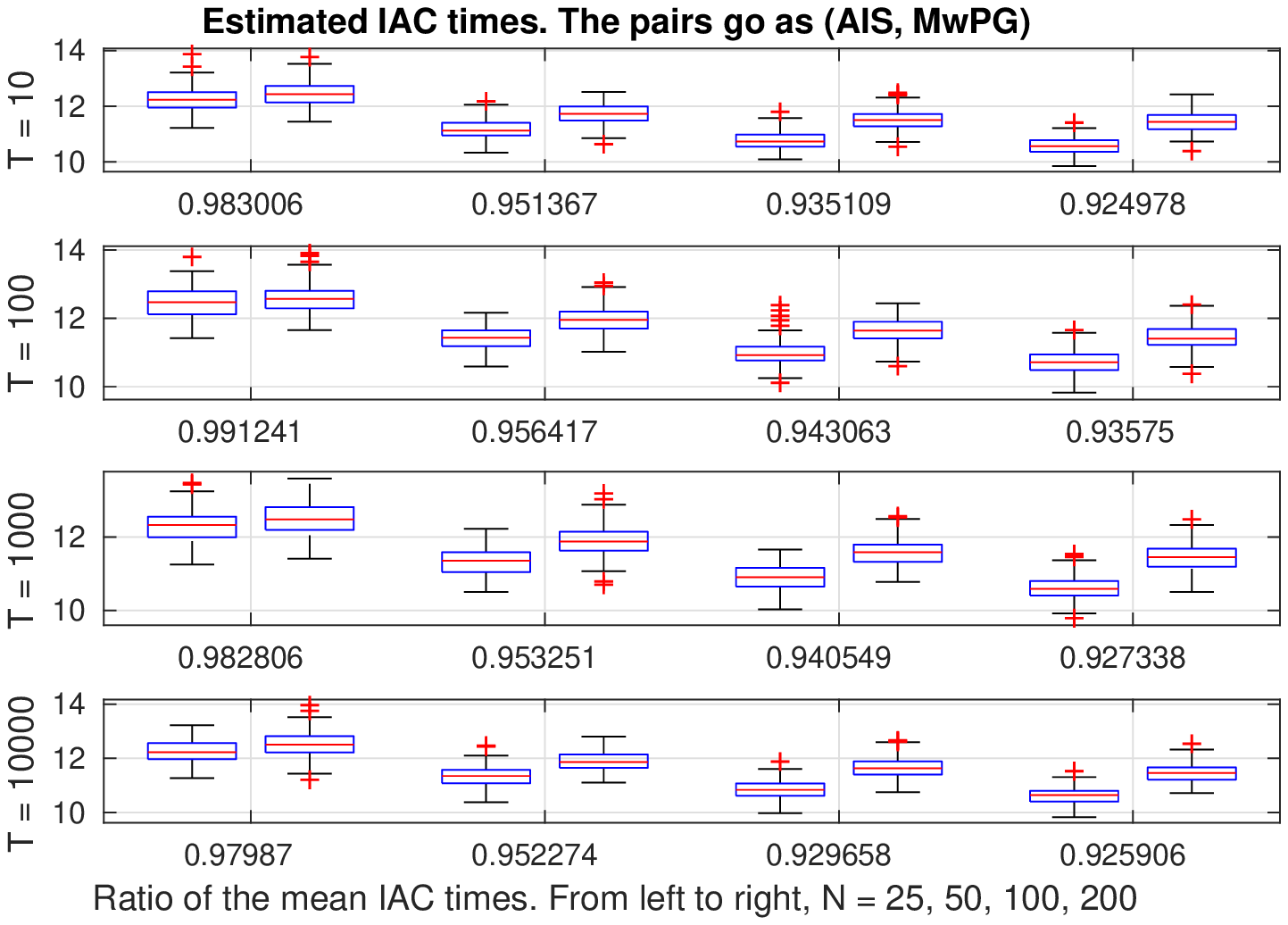}}
\centerline{\includegraphics[scale=0.5]{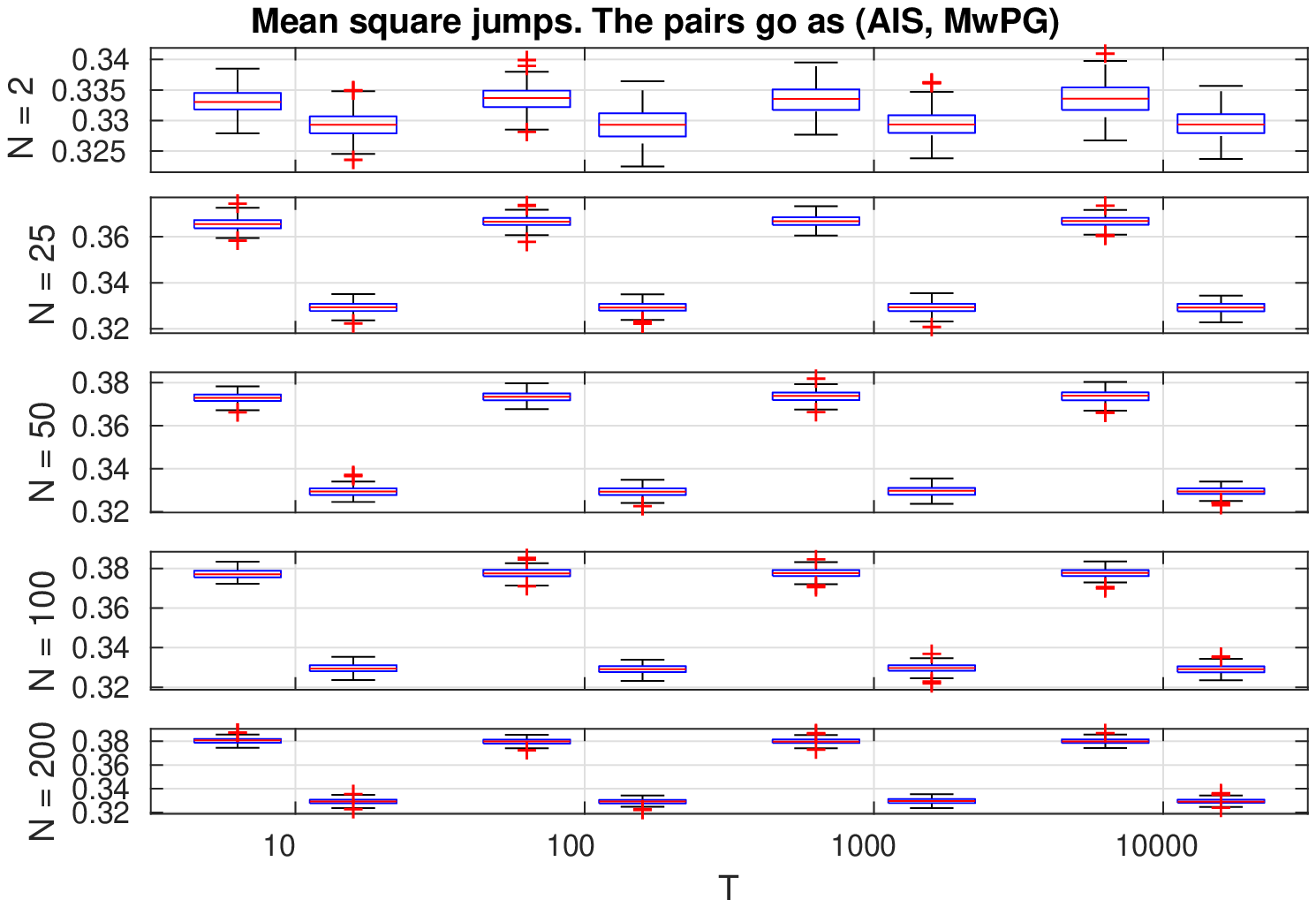}
\hspace{-0.8cm} \includegraphics[scale=0.5]{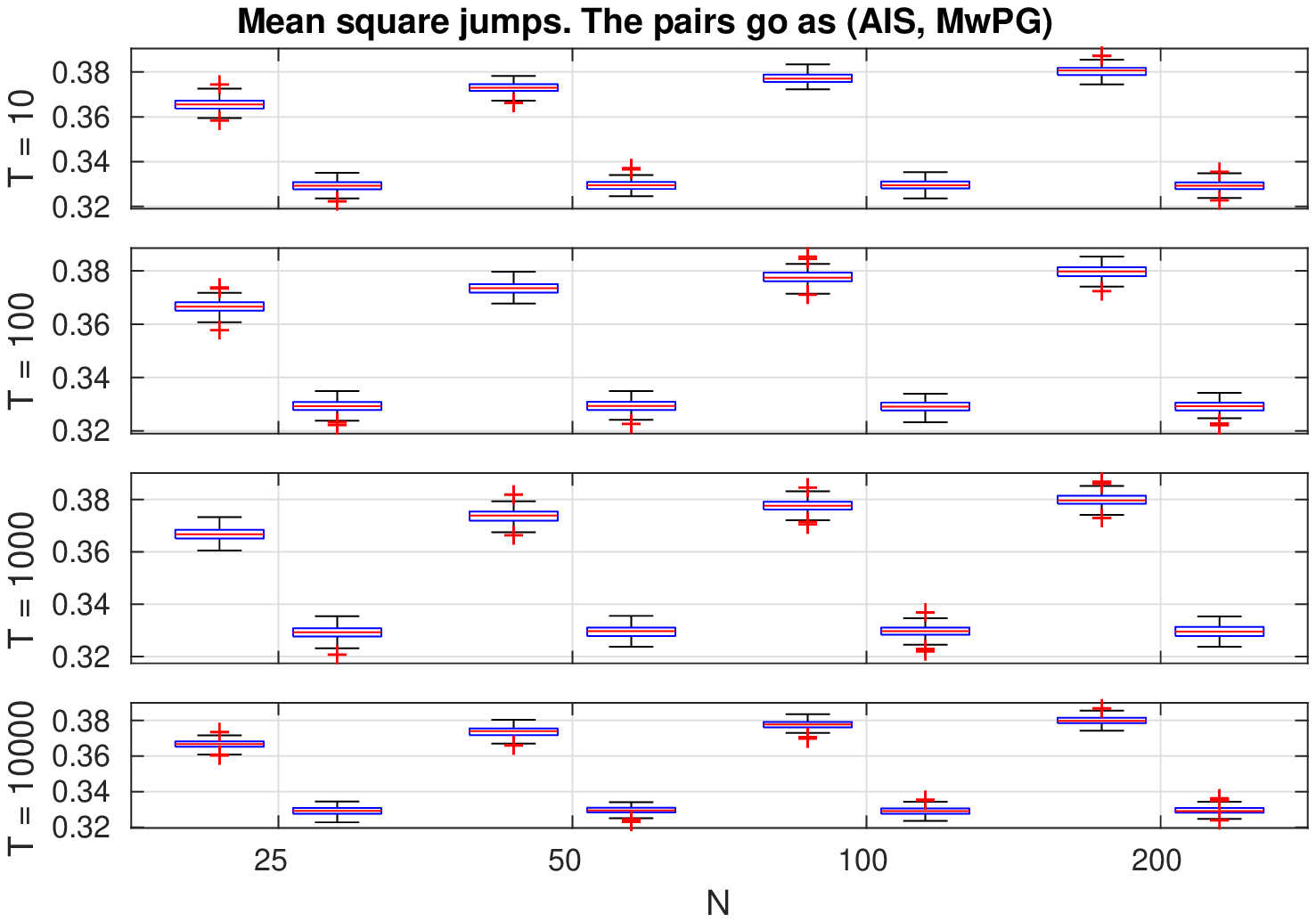}}
\protect\caption{IAC times and mean square jumps (multiplied by $T$) for the hierarchical
model for $a=0.1$. Ratios of the mean IAC times (MCMC AIS cSMC-BS/MwPG)
are shown on the $x-$axis. On the right hand side plots, results
for $N=2$ are not shown to improve readability).}

\label{fig: AIS vs MwPG PNCNHM informative observations N vs T multiple runs ratios included2} 
\end{figure}

\subsection{Experiments on a non-linear state-space model\label{sec: Experiments on a non-linear state-space model} }

We consider now a non-linear SSM often used in the literature to compare
the performance of SMC methods for which $f_{\theta}(x_{t-1},x_{t})=\mathcal{N}\big(x_{t};x_{t-1}/2+25x_{t-1}/(1+x_{t-1}^{2})+8\cos(1.2t),\sigma_{v}^{2}\big)$,
$g_{\theta}(x_{t},y_{t})=\mathcal{N}\big(y_{t};x_{t}^{2}/20,\sigma_{w}^{2}\big)$
and $\mu_{\theta}(x_{1})=\mathcal{N}(x_{1};0,10)$. Here $\theta=(\sigma_{v}^{2},\sigma_{w}^{2})$
and the prior distribution was chosen to be $\sigma_{v}^{2},\sigma_{w}^{2}\overset{\texttt{{\rm iid}}}{\sim}\mathcal{IG}(0.01,0.01)$
where $\mathcal{IG}(a,b)$ is the inverse Gamma distribution with
shape and scale parameters $a$ and $b$. Throughout the experiments,
we generated data using the values $\sigma_{v}^{2}=100$ and $\sigma_{w}^{2}=1$

\subsubsection{Comparison of algorithms for fixed $T$ and varying $N,K$\label{sec: Comparison of Algorithms for fixed T and varying N K} }

We first compare the performance of PMMH, MCMC AIS cSMC, MCMC AIS
cSMC-BS and MwPG for fixed $T=500$ and various values of $K$ and
$N$, for an approximately constant computational budget. To that
purpose, for a given number of intermediate distributions $K$ we
fix the number of particles to $N_{0}=500$ in the cSMC or cSMC-BS
updates used to implement MCMC AIS, while we take the number of particles
to be $N_{0}K$ for both the SMC and cSMC used within the PMMH and
MwPG algorithms respectively. For the MCMC AIS algorithms, the intermediate
distributions are chosen to be of the form $\gamma_{\theta,\theta',k}=\gamma_{\theta_{k}}$,
where $\theta_{k}=(1-\varsigma_{k})\theta+\varsigma_{k}\theta'$,
$\varsigma_{k}=k/(K+1)$, $k=0,\ldots,K+1$. Wherever an SMC or a
cSMC routine is required for the implementation of the algorithms,
multinomial resampling is used at every time step and the transition
density of the SSM is used as the importance sampling distribution.
We used a normal random walk proposal with diagonal covariance matrix
for the RWM updates, where the standard deviation for $\sigma_{v}$
was $0.15$ and $0.08$ for $\sigma_{w}$. We report box plots of
the IAC times associated to $\sigma_{v}^{2}$ and $\sigma_{w}^{2}$
in Figure \ref{fig: IAC times with fixed observation length} and
average IAC times in Table \ref{tbl:_IAC_times_with_fixed_observation_length}.
As observed earlier for the independent scenario the MwPG reaches
its asymptotic regime for small values of $N$ and does not see its
performance improve with the number of particles. This is in contrast
with the PMMH and MCMC AIS cSMC-BS algorithms which achieve similar
performance for large values of $K$ or $N$ and outperform the MwPG
algorithm. We note the crucial role played by the backward sampling
stage in the MCMC AIS algorithm and recall the reader here that the
MwPG also relies on a cSMC-BS step.

\begin{figure}[h]
\centerline{\includegraphics[scale=0.7]{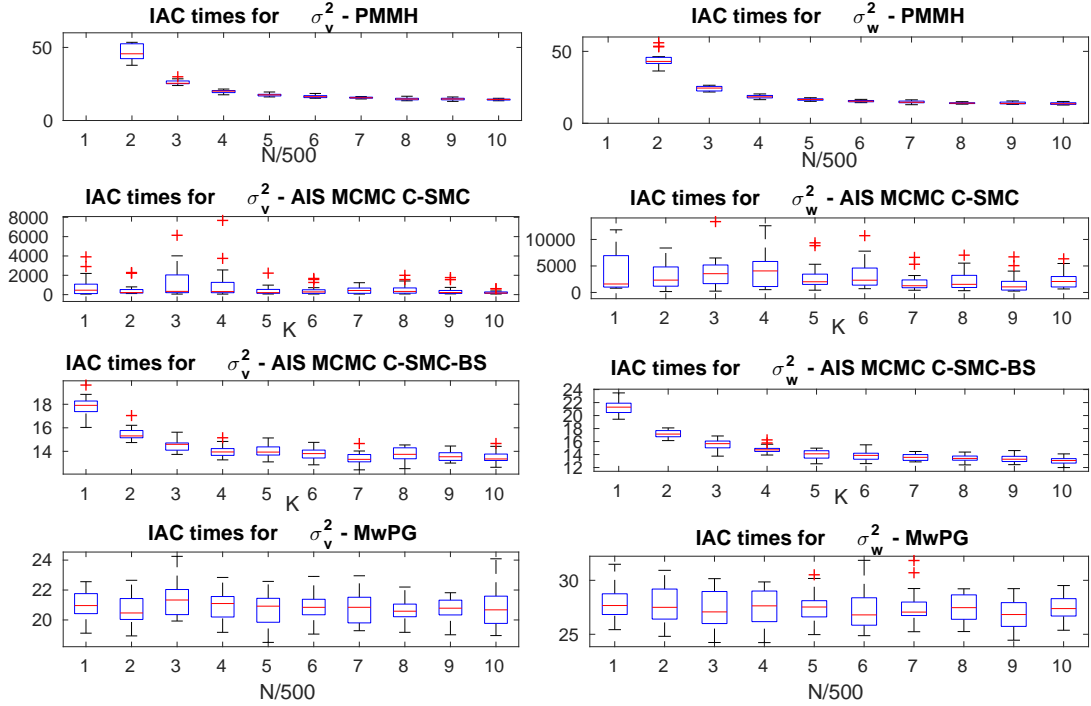}}
\protect\protect\caption{Box plots for the IAC times for $\sigma_{v}^{2}$ and $\sigma_{w}^{2}$
for algorithms PMMH, MCMC AIS cSMC, MCMC AIS cSMC-BS, and MwPG for
various combinations of $N$ and $K$.}

\label{fig: IAC times with fixed observation length} 
\end{figure}

\begin{table}[h]
\centering{}%
\begin{tabular}{|c|c|c|c|c|c|c|c|c|}
\cline{2-9} 
\multicolumn{1}{c|}{} & \multicolumn{2}{c|}{MCMC AIS cSMC} & \multicolumn{2}{c|}{MCMC AIS cSMC-BS} & \multicolumn{2}{c|}{MwPG} & \multicolumn{2}{c|}{PMMH}\tabularnewline
\cline{2-9} 
\multicolumn{1}{c|}{} & $\sigma_{v}^{2}$  & $\sigma_{w}^{2}$  & $\sigma_{v}^{2}$  & $\sigma_{w}^{2}$  & $\sigma_{v}^{2}$  & $\sigma_{w}^{2}$  & $\sigma_{v}^{2}$  & $\sigma_{w}^{2}$ \tabularnewline
\hline 
$K=1$  & 44.9  & 657.2 & 17.7  & 20.9 & 22.9  & 29.8  & 161.9  & 309.3\tabularnewline
$K=2$  & 74.3 & 3096.8  & 14.5  & 15.7 & 22.1 & 28.4  & 41.8 & 43.5 \tabularnewline
$K=3$  & 128.6  & 1960.0 & 13.9 & 15.6  & 22.8  & 28.1 & 22.6  & 21.6 \tabularnewline
$K=4$  & 114.0  & 1428.2 & 15.0  & 15.9  & 20.0  & 31.1  & 19.0  & 19.3\tabularnewline
$K=5$  & 170.8  & 472.2  & 13.4 & 14.9  & 20.4 & 25.8  & 18.9  & 17.5\tabularnewline
$K=6$  & 200.6  & 148.4  & 13.0  & 13.1  & 20.8  & 26.3  & 16.9  & 16.0\tabularnewline
$K=7$  & 66.3 & 1733.6 & 13.7  & 12.4  & 18.3  & 26.5  & 16.6  & 14.1 \tabularnewline
$K=8$  & 638.9 & 544.5 & 13.7  & 12.6  & 22.7  & 27.6 & 14.3  & 13.7\tabularnewline
$K=9$  & 122.2  & 1132.9  & 12.0  & 12.2 & 21.9  & 29.8  & 16.3  & 14.0 \tabularnewline
$K=10$  & 724.6  & 267.3  & 13.5 & 13.7 & 22.7  & 26.7  & 14.9  & 14.0 \tabularnewline
\hline 
\end{tabular}\caption{Estimated IAC times for $\sigma_{v}^{2}$ and $\sigma_{w}^{2}$ for
the algorithms considered. On each row the estimated IAC times for
the MCMC AIS algorithms for $N_{0}=500$ particles and $K$ intermediate
steps and MwG and PMMH algorithms for $N=KN_{0}$ particles are shown.\label{tbl:_IAC_times_with_fixed_observation_length} }
\end{table}

\subsubsection{Comparison of algorithms for fixed $N$ and varying $T$\label{sec: Comparison of algorithms for fixed N and varying T} }

In a second experiment we compared PMMH, MCMC AIS cSMC-BS for $K=1$,
and MwPG for varying values of $T$, in order to assess their scalability
to the size of the observations. All the algorithms used the same
number of particles in order to ensure comparable computational complexity.
Each algorithm was run $200$ times with $N=200$ particles for $T=1000,2000,5000,10000$,
with the exception of the PMMH for which $N=2000$, as otherwise the
estimation of the IAC times was too unreliable, even for $T=1000$.
The prior distribution and the other algorithm settings were similar
to those of subsection \ref{sec: Comparison of Algorithms for fixed T and varying N K}.
In Figure \ref{fig:MwPHIACvsnNLHMM} we report the box plots for the
IAC times estimated from the $200$ runs, while their averages are
reported in Table \ref{tbl:_IAC_times_with_varying_observation_length}.
The PMMH algorithm clearly does not scale well as $T$ increases,
in contrast with MCMC AIS cSMC-BS and MwPG which exhibit remarkable
scaling properties, similar to those observed in the iid scenario.
In line with our earlier findings, MCMC AIS cSMC-BS seems to be consistently
marginally superior to MwPG, for a comparable computational cost.

\begin{figure}[h]
\centerline{\includegraphics[scale=0.7]{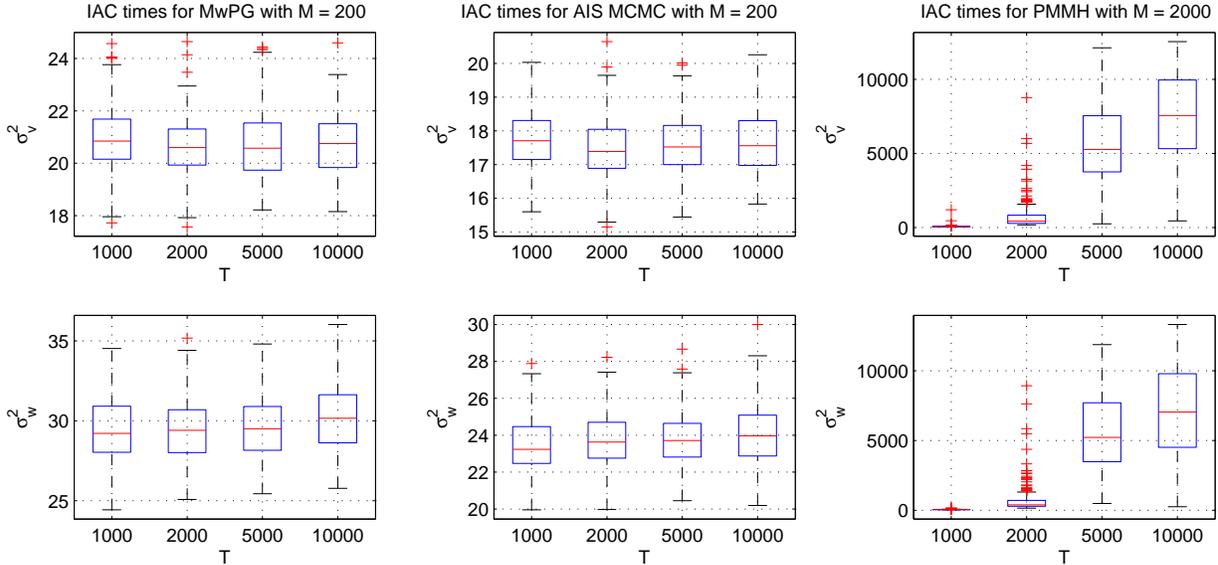}}
\protect\protect\caption{Box plots for the IAC times for $\sigma_{v}^{2}$ and $\sigma_{w}^{2}$
for MCMC AIS cSMC-BS and MwPG with $N=200$ and PMMH with $N=2000$.
Mean IAC values are given in Table \ref{tbl:_IAC_times_with_varying_observation_length}.}

\label{fig:MwPHIACvsnNLHMM} 
\end{figure}

\begin{table}[h]
\centering{}%
\begin{tabular}{|c|c|c|c|c|c|c|}
\cline{2-7} 
\multicolumn{1}{c|}{} & \multicolumn{2}{c|}{MCMC AIS cSMC-BS} & \multicolumn{2}{c|}{MwPG} & \multicolumn{2}{c|}{PMMH}\tabularnewline
\cline{2-7} 
\multicolumn{1}{c|}{} & $\sigma_{v}^{2}$  & $\sigma_{w}^{2}$  & $\sigma_{v}^{2}$  & $\sigma_{w}^{2}$  & $\sigma_{v}^{2}$  & $\sigma_{w}^{2}$ \tabularnewline
\hline 
$T=1000$  & 17.7  & 23.5 & 20.9 & 29.4 & 71.3  & 59.2\tabularnewline
$T=2000$  & 17.5  & 23.7 & 20.6  & 29.4  & 759.0  & 757.9 \tabularnewline
$T=5000$  & 17.6  & 23.7 & 20.7  & 29.6 & 5808.6  & 5663.5 \tabularnewline
$T=10000$  & 17.6 & 24.0 & 20.7  & 30.2  & 7368.1  & 7170.9 \tabularnewline
\hline 
\end{tabular}\caption{Estimated IAC times for $\sigma_{v}^{2}$ and $\sigma_{w}^{2}$ for
MwPG and MCMC AIS cSMC-BS (with $K=1$) for $N=200$ and $N=2000$
for PMMH.\label{tbl:_IAC_times_with_varying_observation_length} }
\end{table}

{\footnotesize{}{}{}{ \linespread{1.0} }}{\footnotesize \par}

\section{Discussion\label{sec:Discussion} }

We have introduced a novel likelihood ratio estimator for SSMs which
relies on an original combination of AIS and cSMC and have shown how
it can be used to obtain an MCMC algorithm to perform Bayesian parameter
inference. In the i.i.d. case, we have provided theory for this estimator
which suggests that the resulting MCMC algorithm has a computational
cost at each iteration scaling linearly with $T$ instead of quadratically
for standard pseudo-marginal methods. In the general SSM case, we
conjecture that similar results also hold for the class of state-space
models where cSMC-BS is efficient as evidenced by our empirical results.

\section*{Acknowledgements}

Arnaud Doucet's research is supported by the Engineering and Physical
Sciences Research Council (EPSRC) EP/K000276/1 Advanced Monte Carlo
Methods for Inference in Complex Dynamic Models and EP/K009850/1 Bayesian
Inference for Big Data with Stochastic Gradient Markov Chain Monte
Carlo. Christophe Andrieu's research was supported by EPSRC EP/K009575/1
Bayesian Inference for Big Data with Stochastic Gradient Markov Chain
Monte Carlo and EP/K0\-14463/1 Intractable Likelihood: New Challenges
from Modern Applications (ILike). Sinan Y\i ld\i r\i m's research
was also supported by ILike, EPSRC EP/K0\-14463/1. The authors acknowledge
the (intensive) use of the Blue Crystal HPC facility at the University
of Bristol.

{\small{}\bibliographystyle{Chicago}
\bibliography{myrefs_thesis}
}{\small \par}

{\footnotesize{}\newpage{}}{\footnotesize \par}

\appendix
{\footnotesize{}

\section{Approximation}

We first establish a simple approximation of $\Lambda_{T}(\theta,\epsilon;\omega,\xi)$,
which relies on a Taylor expansion. We let $\mathring{\Theta}$ be
the interior of $\Theta$. 
\begin{lem}
\label{lem:taylor-approximation}Assume (A\ref{hyp:strong-likelihood-cond}).
For any $(\theta,\epsilon,\xi,\omega)\in\mathring{\Theta}\times\Xi\times\mathsf{X}^{\mathbb{N}}\times\mathsf{Y}^{\mathbb{N}}$
and any $T\in\mathbb{N}$ such that $\tilde{\theta}(\epsilon,T)\in\mathring{\Theta}$
there exist $\{\mathring{\theta}_{t},1\leq t\leq T\},\{\bar{\theta}_{t},1\leq t\leq T\}\in[\theta\wedge\theta{}^{'}(\epsilon,T),\theta\vee\theta{}^{'}(\epsilon,T)]^{T}$
such that

\[
\Lambda_{T}(\theta,\epsilon;\omega,\xi)=S_{\theta,\epsilon,T}^{(1)}(\omega,\xi)+S_{\theta,\epsilon,T}^{(2)}(\omega,\xi)+S_{\theta,\epsilon,T}^{(3)}(\omega,\xi)\quad.
\]
with 
\[
S_{\theta,\epsilon,T}^{(1)}(\omega,\xi):=\frac{\epsilon}{2\sqrt{T}}\sum_{t=1}^{T}\left\{ \dot{\ell}_{\theta}\big(x_{t}\mid y_{t}\big)+\dot{\ell}_{\tilde{\theta}(\epsilon,T)}\big(x'_{t}\mid y_{t}\big)\right\} ,
\]
\[
S_{\theta,\epsilon,T}^{(2)}(\omega,\xi):=\frac{\epsilon^{2}}{8T}\sum_{t=1}^{T}\left\{ \ddot{\ell}_{\theta}(x_{t}\mid y_{t})+\ddot{\ell}_{\tilde{\theta}(\epsilon,T)}(x'_{t}\mid y_{t})\right\} ,
\]
\[
S_{\theta,\epsilon,T}^{(3)}(\omega,\xi):=\frac{\epsilon^{3}}{48T\sqrt{T}}\sum_{t=1}^{T}\left\{ \dddot{\ell}_{\bar{\theta}_{t}}(x_{t}\mid y_{t})+\dddot{\ell}_{\mathring{\theta}_{t}}(x'_{t}\mid y_{t})\right\} .
\]
\end{lem}
\begin{proof}
Recall that 
\begin{align*}
\Lambda_{T}(\theta,\epsilon;\omega,\xi) & =\sum_{t=1}^{T}\log\frac{p_{\tilde{\theta}(\epsilon,T)}\left(x_{t}\mid y_{t}\right)}{p_{\theta}\left(x_{t}\mid y_{t}\right)}+\log\frac{p_{\theta^{\prime}(\epsilon,T)}\left(x'_{t}\mid y_{t}\right)}{p_{\tilde{\theta}(\epsilon,T)}\left(x'_{t}\mid y_{t}\right)}.
\end{align*}
For $(\theta,\tilde{\theta})\in\mathring{\Theta}$ and $(x,y)\in\mathsf{X}\times\mathsf{Y}$
a Taylor expansion yields 
\begin{align*}
\log\frac{p_{\tilde{\theta}}\left(x\mid y\right)}{p_{\theta}\left(x\mid y\right)} & =\dot{\ell}_{\theta}(x\mid y)\big(\tilde{\theta}-\theta\big)+\frac{1}{2}\ddot{\ell}_{\theta}(x\mid y)\big(\tilde{\theta}-\theta\big)^{2}+\frac{1}{6}\dddot{\ell}_{\bar{\theta}}(x\mid y)\big(\tilde{\theta}-\theta\big)^{3}
\end{align*}
for some $\bar{\theta}\in[\tilde{\theta}\wedge\theta^{\prime},\tilde{\theta}\vee\theta^{\prime}]$,
also dependent on $x$ and $y$. Similarly for $(\tilde{\theta},\theta')\in\mathring{\Theta}$
\begin{align*}
\log\frac{p_{\theta^{\prime}}\left(x\mid y\right)}{p_{\tilde{\theta}}\left(x\mid y\right)} & =\dot{\ell}_{\tilde{\theta}}(x\mid y)\big(\theta^{\prime}-\tilde{\theta}\big)+\frac{1}{2}\ddot{\ell}_{\tilde{\theta}}(x\mid y)\big(\theta^{\prime}-\tilde{\theta}\big)^{2}+\frac{1}{6}\dddot{\ell}_{\mathring{\theta}}(x\mid y)\big(\theta^{\prime}-\tilde{\theta}\big)^{3}
\end{align*}
for some $\mathring{\theta}\in[\tilde{\theta}\wedge\theta^{\prime},\tilde{\theta}\vee\theta^{\prime}]$,
also dependent on $x$ and $y$. It follows that 
\begin{align*}
\Lambda_{T}(\theta,\epsilon;\omega,\xi) & =\sum_{t=1}^{T}\dot{\ell}_{\theta}(x_{t}\mid y_{t})\big(\tilde{\theta}(\epsilon,T)-\theta\big)+\dot{\ell}_{\tilde{\theta}(\epsilon,T)}(x'_{t}\mid y_{t})\big(\theta'(\epsilon,T)-\tilde{\theta}(\epsilon,T)\big)\\
 & +\frac{1}{2}\sum_{t=1}^{T}\ddot{\ell}_{\theta}(x_{t}\mid y_{t})\big(\tilde{\theta}(\epsilon,T)-\theta\big)^{2}+\ddot{\ell}_{\tilde{\theta}(\epsilon,T)}(x'_{t}\mid y_{t})\big(\theta'(\epsilon,T)-\tilde{\theta}(\epsilon,T)\big)^{2}\\
 & +\frac{1}{6}\sum_{t=1}^{T}\dddot{\ell}_{\bar{\theta}_{t}}(x_{t}\mid y_{t})\big(\tilde{\theta}(\epsilon,T)-\theta\big)^{3}+\dddot{\ell}_{\mathring{\theta}_{t}}(x'_{t}\mid y_{t})\big(\theta'(\epsilon,T)-\tilde{\theta}(\epsilon,T)\big)^{3}
\end{align*}
Now, from the definition of $\tilde{\theta}(\epsilon,T)$ we have
\[
\tilde{\theta}(\epsilon,T)-\theta:=\frac{\epsilon/2}{\sqrt{T}},\text{ }\theta^{\prime}(\epsilon,T)-\tilde{\theta}(\epsilon,T):=\frac{\epsilon/2}{\sqrt{T}},
\]
and therefore 
\begin{align*}
\Lambda_{T}(\theta,\epsilon;\omega,\xi) & =\frac{\epsilon}{2\sqrt{T}}\sum_{t=1}^{T}\left\{ \dot{\ell}_{\theta}(x_{t}\mid y_{t})+\dot{\ell}_{\tilde{\theta}(\epsilon,T)}(x'_{t}\mid y_{t})\right\} \\
 & +\frac{\epsilon^{2}}{8T}\sum_{t=1}^{T}\left\{ \ddot{\ell}_{\theta}(x_{t}\mid y_{t})+\ddot{\ell}_{\tilde{\theta}(\epsilon,T)}(x'_{t}\mid y_{t})\right\} \\
 & +\frac{\epsilon^{3}}{48T\sqrt{T}}\sum_{t=1}^{T}\left\{ \dddot{\ell}_{\bar{\theta}_{t}}(x_{t}\mid y_{t})+\dddot{\ell}_{\mathring{\theta}_{t}}(x'_{t}\mid y_{t})\right\} .
\end{align*}
\end{proof}
This is a purely technical lemma to establish various continuity properties
needed after. 
\begin{lem}
\label{lem:basicLipschitzcontinuities}Assume (A\ref{hyp:strong-likelihood-cond}).
Then for any $(\theta,\theta')\in\mathring{\Theta}^{2}$, 
\[
\sup_{y\in\mathsf{Y}}\|p_{\theta}(\cdot\mid y)-p_{\theta'}(\cdot\mid y)\|_{tv}\leq\frac{1}{2}\bar{\ell}^{(1)}|\theta-\theta'|.
\]
Let $\phi_{\theta}(\cdot,\cdot):\Theta\times\mathsf{X}\times\mathsf{Y}\rightarrow\mathbb{R}$
and define $\bar{\phi}:=\sup_{(\theta,x,y)\in\Theta\times\mathsf{X}\times\mathsf{Y}}|\phi_{\theta}(x,y)|$.
Then for any $(\theta,\epsilon,\omega)\in\Theta\times\Xi\times\mathsf{Y}^{\mathbb{N}}$,
$T\geq1$ and $N\in\mathbb{N}$ 
\begin{multline*}
\left|\mathbb{E}_{\theta,\epsilon,T}^{\omega}\left[\phi_{\theta}(X{}_{1},y)\right]-\mathbb{E}_{\theta,\epsilon,T}^{\omega}\left[\phi_{\tilde{\theta}(\epsilon,T)}(X_{1}',y)\right]\right|\leq\bar{\phi}\bar{\ell}^{(1)}\times\frac{|\epsilon|}{\sqrt{T}}\sup_{(\theta,x,y)\in\Theta\times\mathsf{X}\times\mathsf{Y}}\|R_{\theta,y}^{[N]}\big(x,\cdot\big)-p_{\theta}(\cdot\mid y)\|_{tv}\\
+\left|\mathbb{E}_{\tilde{\theta}(\epsilon,T)}^{\omega}\left[\phi_{\tilde{\theta}(\epsilon,T)}(X{}_{1},y)\right]-\mathbb{E}_{\theta}^{\omega}\left[\phi_{\theta}(X{}_{1},y)\right]\right|.
\end{multline*}
If in addition there exists $\tilde{\phi}>0$ such that for all $(\theta,\theta',x,y)\in\mathring{\Theta}^{2}\times\mathsf{X}\times\mathsf{Y}$,
$|\phi_{\theta}(x,y)-\phi_{\theta'}(x,y)|\leq\tilde{\phi}|\theta-\theta'|$
then 
\[
\left|\mathbb{E}_{\theta'}^{\omega}\left[\phi_{\theta'}(X{}_{1},y)\right]-\mathbb{E}_{\theta}^{\omega}\left[\phi_{\theta}(X{}_{1},y)\right]\right|\leq(\bar{\phi}\bar{\ell}^{(1)}+\tilde{\phi})|\theta-\theta'|.
\]
\end{lem}
\begin{proof}
We have for any $y\in\mathsf{Y}$ 
\begin{align*}
\|p_{\theta}(\cdot & \mid y)-p_{\theta'}(\cdot\mid y)\|_{tv}=\frac{1}{2}\int_{\mathsf{X}}|\exp(\ell_{\theta}(x\mid y))-\exp(\ell_{\theta'}(x\mid y))|{\rm d}x\\
 & =\frac{1}{2}\int_{\mathsf{X}}|\int_{\theta}^{\theta'}\dot{\ell}_{\vartheta}(x\mid y)\exp(\ell_{\vartheta}(x\mid y)){\rm d}\vartheta|{\rm d}x\\
 & \leq\frac{1}{2}\bar{\ell}^{(1)}|\int_{\theta}^{\theta'}\int_{\mathsf{X}}\exp(\ell_{\vartheta}(x\mid y)){\rm d}x{\rm d}\vartheta|\\
 & =\frac{1}{2}\bar{\ell}^{(1)}|\theta-\theta'|.
\end{align*}
For the next statement we use standard operator notation for brevity:
for a probability distribution $\mu$, a Markov operator $\Pi$ and
a function $f$, we let $\Pi f(x):=\int f(u)\Pi(x,{\rm d}u)$ and
$\mu f=\mu(f):=\int f(u)\mu({\rm d}u)$. We have the decomposition,
for $\theta,\tilde{\theta}\in\Theta$ and $N\in\mathbb{N}$ 
\begin{align*}
p_{\theta}R_{\tilde{\theta},y}^{[N]}\left(\phi_{\tilde{\theta}}\right)-p_{\theta}\left(\phi_{\theta}\right) & =p_{\theta}R_{\tilde{\theta},y}^{[N]}\left(\phi_{\tilde{\theta}}\right)-p_{\tilde{\theta}}R_{\tilde{\theta},y}^{[N]}\left(\phi_{\tilde{\theta}}\right)+p_{\tilde{\theta}}\left(\phi_{\tilde{\theta}}\right)-p_{\theta}\left(\phi_{\theta}\right)\\
 & =(p_{\theta}-p_{\tilde{\theta}})R_{\tilde{\theta},y}^{[N]}(\phi_{\tilde{\theta}}-p_{\tilde{\theta}}\phi_{\tilde{\theta}})+(p_{\tilde{\theta}}-p_{\theta})\left(\phi_{\tilde{\theta}}\right)-p_{\theta}(\phi_{\theta}-\phi_{\tilde{\theta}})
\end{align*}
and 
\begin{align*}
|(p_{\theta}-p_{\tilde{\theta}})R_{\tilde{\theta},y}^{[N]}(\phi_{\tilde{\theta}}-p_{\tilde{\theta}}\phi_{\tilde{\theta}})| & \leq2\|p_{\theta}-p_{\tilde{\theta}}\|_{tv}\sup_{x\in\mathsf{X}}|R_{\tilde{\theta},y}^{[N]}(\phi_{\tilde{\theta}}-p_{\tilde{\theta}}\phi_{\tilde{\theta}})(x)|\\
 & \leq2\|p_{\theta}-p_{\tilde{\theta}}\|_{tv}2\sup_{x\in\mathsf{X}}\|R_{\tilde{\theta},y}^{[N]}(x,\cdot)-p_{\tilde{\theta}}(\cdot\mid y)\|_{tv}\bar{\phi}.
\end{align*}
Finally we have the decomposition and bound for $\theta,\theta'\in\mathring{\Theta}$
\begin{align*}
\left|(p_{\theta'}-p_{\theta})\left(\phi_{\theta'}\right)-p_{\theta}(\phi_{\theta}-\phi_{\theta'})\right| & \leq2\bar{\phi}\|p_{\theta}(\cdot\mid y)-p_{\theta'}(\cdot\mid y)\|_{tv}+\tilde{\phi}|\theta-\theta'|\\
 & =(\bar{\phi}\bar{\ell}^{(1)}+\tilde{\phi})|\theta-\theta'|.
\end{align*}
\end{proof}
We establish a first level of approximation of $\Lambda_{T}(\theta,\epsilon;\omega,\xi)$
in the following sense. 
\begin{lem}
\label{lem:firstapproximation}Assume (A\ref{hyp:strong-likelihood-cond}).
For any $(\theta,\epsilon,\omega,T)\in\Theta\times\Xi\times\mathsf{Y}^{\mathbb{N}}\times\mathbb{N}$,
let 
\[
\bar{S}_{\theta,T}^{(2)}(\omega):=\frac{\epsilon^{2}}{4T}\sum_{t=1}^{T}\mathbb{E}_{\theta}^{\omega}\big[\ddot{\ell}_{\theta}(X_{t}\mid y_{t})\big].
\]
Then for any $\omega\in\mathsf{Y}^{\mathbb{N}}$ and with the notation
of Lemma \ref{lem:taylor-approximation}, 
\[
\lim_{T\rightarrow\infty}\sup_{(N_{T},\theta,\epsilon)\in\mathbb{N}\times\mathring{\Theta}\times\Xi}\mathbb{E}_{\theta,\epsilon,T}^{\omega}\big|\Lambda_{T}(\theta,\epsilon;\omega,\xi)-S_{\theta,\epsilon,T}^{(1)}(\omega,\xi)-\bar{S}_{\theta,\epsilon,T}^{(2)}(\omega)\big|=0.
\]
\end{lem}
\begin{proof}
First we have 
\[
\Lambda_{T}(\theta,\epsilon)-S_{\theta,\epsilon,T}^{(1)}(\omega,\xi)-\bar{S}_{\theta,\epsilon,T}^{(2)}(\omega,\xi)=S_{\theta,\epsilon,T}^{(2)}(\omega,\xi)-\bar{S}_{\theta,\epsilon,T}^{(2)}(\omega)+S_{\theta,\epsilon,T}^{(3)}(\omega,\xi)
\]
and we are going to consider the second order moment of the term on
the right hand side\textendash we will then invoke the standard inequality
${\rm \mathbb{E}}_{\theta,\epsilon,T}^{\omega}|Z|\leq\sqrt{{\rm \mathbb{E}}_{\theta,\epsilon,T}^{\omega}[Z^{2}]}$
in order to conclude. In order to alleviate notation we introduce
$\|Z\|_{2}:=\sqrt{{\rm \mathbb{E}}_{\theta,\epsilon,T}^{\omega}\left(Z^{2}\right)}$,
which satisfies the triangle inequality, and drop the dependence on
$\omega$. We bound $\|S_{\theta,\epsilon,T}^{(3)}\|_{2}$ and $\|S_{\theta,\epsilon,T}^{(2)}-\bar{S}_{\theta,\epsilon,T}^{(2)}\|_{2}$.
Clearly we have 
\[
\|S_{\theta,\epsilon,T}^{(3)}\|_{2}\leq\frac{|\epsilon|^{3}}{48\sqrt{T}}2\bar{\ell}^{(3)}\quad.
\]
Now define 
\[
\tilde{S}_{\theta,\epsilon,T}^{(2)}:=\frac{\epsilon^{2}}{8T}\sum_{t=1}^{T}\left\{ \mathbb{E}_{\theta,\epsilon,T}^{\omega}[\ddot{\ell}_{\theta}(X_{t}\mid y_{t})+\ddot{\ell}_{\tilde{\theta}(\epsilon,T)}(X'_{t}\mid y_{t})]\right\} 
\]
and consider the upper bound 
\[
\|S_{\theta,\epsilon,T}^{(2)}-\bar{S}_{\theta,\epsilon,T}^{(2)}\|_{2}\leq\|S_{\theta,\epsilon,T}^{(2)}-\tilde{S}_{\theta,\epsilon,T}^{(2)}\|_{2}+\|\tilde{S}_{\theta,\epsilon,T}^{(2)}-\bar{S}_{\theta,\epsilon,T}^{(2)}\|_{2}.
\]
Using independence we obtain 
\begin{align*}
\|S_{\theta,\epsilon,T}^{(2)}-\tilde{S}_{\theta,\epsilon,T}^{(2)}\|_{2} & =\frac{\epsilon^{2}}{8T}\sqrt{\sum_{t=1}^{T}{\rm \mathbb{V}}_{\theta,\epsilon,T}^{\omega}\left(\ddot{\ell}_{\theta}(X_{t}\mid y_{t})+\ddot{\ell}_{\tilde{\theta}(\epsilon,T)}(X'_{t}\mid y_{t})\right)}\\
 & \leq\frac{\epsilon^{2}}{4\sqrt{T}}\bar{\ell}^{(2)}.
\end{align*}
Finally 
\begin{align*}
\|\tilde{S}_{\theta,\epsilon,T}^{(2)}-\bar{S}_{\theta,\epsilon,T}^{(2)}\|_{2} & =\frac{\epsilon^{2}}{8T}\|\sum_{t=1}^{T}\left\{ \mathbb{E}_{\theta,\epsilon,T}^{\omega}[\ddot{\ell}_{\theta}(X_{t}\mid y_{t})-\ddot{\ell}_{\tilde{\theta}(\epsilon,T)}(X'_{t}\mid y_{t})]\right\} \|_{2}.
\end{align*}
The estimate of the difference obtained in Lemma \ref{lem:basicLipschitzcontinuities}
leads to 
\begin{multline*}
|\mathbb{E}_{\theta,\epsilon,T}^{\omega}[\ddot{\ell}_{\theta}(X_{t}\mid y_{t})-\ddot{\ell}_{\tilde{\theta}(\epsilon,T)}(X'_{t}\mid y_{t})]|\\
\leq\bar{\ell}^{(1)}\bar{\ell}^{(2)}\frac{|\epsilon|}{\sqrt{T}}\sup_{(\theta,x,y)\in\mathring{\Theta}\times\mathsf{X}\times\mathsf{Y}}\|R_{\theta,y}^{[N_{T}]}\big(x,\cdot\big)-p_{\theta}(\cdot\mid y)\|_{tv}+(\bar{\ell}^{(1)}\bar{\ell}^{(2)}+\bar{\ell}^{(3)})\frac{|\epsilon|}{2\sqrt{T}}
\end{multline*}
and we conclude since the total variation term is bounded by $1$.
\end{proof}
The following result establishes that $P-a.s.$ one can approximate
$\Lambda_{T}(\theta,\epsilon;\omega,\xi)$ with $S_{\theta,\epsilon,T}^{(1)}(\omega,\xi)-\sigma^{2}(\theta,\epsilon)/2$
in the sense given in the corollary below. 
\begin{lem}
\label{lem:ucvbarS2}Assume (A\ref{hyp:strong-likelihood-cond}),
then 
\[
\lim_{T\rightarrow\infty}\sup_{(N_{T},\theta,\epsilon)\in\mathbb{N}\times\Theta\times\Xi}\big|\bar{S}_{\theta,\epsilon,T}^{(2)}(\omega)-\frac{\epsilon^{2}}{4}\mathbb{E_{\theta}}\big[\ddot{\ell}_{\theta}(X_{1}\mid Y_{1})\big]\big|=0\quad P-a.s.
\]
\end{lem}
\begin{proof}
We use a straightforward adaptation of the simple result of \citet[Lemma 1]{tauchen1985diagnostic}.
Conditions (iii) and (iv) of \citet[Lemma 1]{tauchen1985diagnostic}
are immediate since for any $\omega\in\mathsf{Y}^{\mathbb{N}}$, $(\theta,\epsilon)\mapsto\epsilon^{2}\mathbb{E}_{\theta}^{\omega}\big[\ddot{\ell}_{\theta}\big(X{}_{1}\mid y\big)\big]$
is continuous from Lemma \ref{lem:basicLipschitzcontinuities} and
$\Theta\times\Xi$ is assumed compact, implying $\sup_{(\theta,\epsilon)\in\Theta\times\Xi}|\epsilon^{2}\mathbb{E}_{\theta}^{\omega}\big[\ddot{\ell}_{\theta}\big(X{}_{1}\mid y\big)\big]|\leq\ddot{\ell}^{(2)}\sup_{\epsilon\in\Xi}\epsilon^{2}<\infty$,
which is obviously integrable w.r.t the distribution of the observations.
We are left with establishing the measurability of the suprema considered,
covered by (ii) of \citet[Lemma 1]{tauchen1985diagnostic}. Note that
if for any $y_{1:T}\in\mathsf{Y}^{T}$ $(\theta,\epsilon)\mapsto\phi(\theta,\epsilon,y_{1:T})$
is continuous then 
\[
y_{1:T}\mapsto\sup_{(\theta,\epsilon)\in\Theta\times\Xi}\phi(\theta,\epsilon,y_{1:T})=\sup_{(\theta,\epsilon)\in(\Theta\times\Xi)\cap\mathbb{Q}^{2}}\phi(\theta,\epsilon,y_{1:T})
\]
is measurable. Since for any $y\in\mathsf{Y}$, $(\theta,\epsilon)\mapsto\epsilon^{2}\mathbb{E}_{\theta}^{\omega}\big[\ddot{\ell}_{\theta}\big(X{}_{1}\mid y\big)\big]$
is continuous by Lemma \ref{lem:basicLipschitzcontinuities}, we conclude.
\end{proof}
\begin{cor}
\label{cor:almostsureapproximation}Recalling that $\sigma^{2}(\theta,\epsilon):=\frac{-\epsilon^{2}}{2}\mathbb{E}_{\theta}\left[\ddot{\ell}_{\theta}(X_{1}\mid Y_{1})\right]$,
$P-a.s.$ we have 
\[
\lim_{T\rightarrow\infty}\sup_{(N_{T},\theta,\epsilon)\in\mathbb{N}\times\mathring{\Theta}\times\Xi}\mathbb{E}_{\theta,\epsilon,T}^{\omega}\big|\Lambda_{T}(\theta,\epsilon;\omega,\xi)-S_{\theta,\epsilon,T}^{(1)}(\omega,\xi)+\sigma^{2}(\theta,\epsilon)/2\big|=0.
\]
\end{cor}
We now seek to establish that $S_{\theta,\epsilon,T}^{(1)}(\omega,\xi)$
satisfies a $(\theta,\epsilon)$-uniform central limit theorem (U-CLT)
with limiting mean and variance $P-a.s.$ independent of $\omega$.

\section{{\normalsize{}{}Conditional CLT for $S_{\theta,\epsilon,T}^{(1)}(\xi,\omega)$}}

We now apply a CLT conditional upon the observations and will notice
that $P-$a.s. the constants involved are asymptotically independent
of the realisation of the observations.

\subsection{{\normalsize{}{}Checking Lyapunov's theorem conditions}}

We will use the following technical lemma 
\begin{lem}
\label{lem:boundcovariance}Assume (A\ref{hyp:strong-likelihood-cond}).
Then for any $(\theta,\epsilon)\in\mathring{\Theta}\times\Xi$, $T\geq1$,
$N_{T}\in\mathbb{N}$ and $\omega\in\mathsf{Y}^{\mathbb{N}}$
\[
\left|\mathbb{C}_{\theta,\epsilon,T}^{\omega}\left[\dot{\ell}_{\theta}(X_{t}\mid y_{t}),\dot{\ell}_{\tilde{\theta}(\epsilon,T)}(X'_{t}\mid y_{t})\right]\right|\leq\sup_{(\theta,x,y)\in\Theta\times\mathsf{X}\times\mathsf{Y}}\|R_{\theta,y}^{[N_{T}]}(x,\cdot)-p_{\theta}(\cdot\mid y)\|_{tv}\times\big(\bar{\ell}^{(1)}\big)^{2}\left[1+\bar{\ell}^{(1)}\times\frac{|\epsilon|}{\sqrt{T}}\right].
\]
Further for any $t\geq1$ we have 
\[
\lim_{T\rightarrow\infty}\sup_{(N_{T},\theta,\epsilon,y_{t})\in\mathbb{N}\times\mathring{\Theta}\times\Xi\times\mathsf{Y}}\left|\mathbb{V}_{\theta,\epsilon,T}^{\omega}\left[\dot{\ell}_{\tilde{\theta}(\epsilon,T)}\big(X'_{t}\mid y_{t}\big)\right]-\mathbb{V}_{\theta}^{y}\left[\dot{\ell}_{\theta}(X{}_{t}\mid y_{t})\right]\right|=0.
\]
\end{lem}
\begin{proof}
First note that $\mathbb{E}_{\tilde{\theta}(\epsilon,T)}^{\omega}\big[\dot{\ell}_{\tilde{\theta}(\epsilon,T)}(X{}_{t}\mid y_{t})\big]=0$
and $\mathbb{E}_{\theta}^{\omega}\left[\dot{\ell}_{\theta}(X_{t}\mid y_{t})\right]=0$
and apply the Cauchy-Schwartz inequality to obtain 
\begin{multline}
\left|\mathbb{E}_{\theta,\epsilon,T}^{\omega}\left[\dot{\ell}_{\theta}(X_{t}\mid y_{t})\mathbb{E}_{\theta,\epsilon,T}^{\omega}\big[\dot{\ell}_{\tilde{\theta}(\epsilon,T)}(X'_{t}\mid y_{t})-\mathbb{E}_{\tilde{\theta}(\epsilon,T)}^{\omega}\big[\dot{\ell}_{\tilde{\theta}(\epsilon,T)}(X{}_{t}\mid y_{t})\big]\mid X_{t}\big]\right]\right|\leq\\
\sqrt{\sup_{(\theta,y_{t})\in\mathring{\Theta}\times\mathsf{Y}}\mathbb{V}_{\theta}^{\omega}\left[\dot{\ell}_{\theta}(X{}_{t}\mid y_{t})\right]}\times2\bar{\ell}^{(1)}\sup_{(\theta,x,y)\in\mathring{\Theta}\times\mathsf{X}\times\mathsf{Y}}\|R_{\theta,y}^{[N_{T}]}(x,\cdot)-p_{\theta}(\cdot\mid y)\|_{tv}.\label{eq:covariancevanishes}
\end{multline}
For the second statement, it is sufficient to show that for $\gamma\in\{1,2\}$
\[
\lim_{T\rightarrow\infty}\sup_{(\theta,\epsilon,y_{t})\in\mathring{\Theta}\times\Xi\times\mathsf{Y}}\left|\mathbb{E}_{\theta}^{\omega}\left[\dot{\ell}_{\tilde{\theta}(\epsilon,T)}^{\gamma}(X'_{t}\mid y_{t})\right]-\mathbb{E}_{\theta}^{\omega}\left[\dot{\ell}_{\theta}^{\gamma}(X{}_{t}\mid y_{t})\right]\right|=0.
\]
The case $\gamma=1$ is treated in the proof of Lemma \ref{lem:limitexpectationinCLT}.
For $\gamma=2$ we again use Lemma \ref{lem:basicLipschitzcontinuities}
and get the bound 
\begin{multline*}
|\mathbb{E}_{\theta,\epsilon,T}^{\omega}[\dot{\ell}_{\theta}^{2}(X_{t}\mid y_{t})-\dot{\ell}_{\tilde{\theta}(\epsilon,T)}^{2}(X'_{t}\mid y_{t})]|\\
\leq\big(\bar{\ell}^{(1)}\big)^{3}\frac{|\epsilon|}{\sqrt{T}}\sup_{(\theta,x,y)\Theta\times\mathsf{X}\times\mathsf{Y}}\|R_{\theta,y}^{[N_{T}]}\big(x,\cdot\big)-p_{\theta}(\cdot\mid y)\|_{tv}+(\big(\bar{\ell}^{(1)}\big)^{3}+2\bar{\ell}^{(1)}\bar{\ell}^{(2)})\frac{|\epsilon|}{2\sqrt{T}}.
\end{multline*}
\end{proof}
\begin{cor}
\label{cor:existenceN0}Under (A\ref{hyp:strong-likelihood-cond})
there exists $N_{0},T_{0}\in\mathbb{N}$ such that for any $\big\{ N_{T}\big\}$
such that $\liminf_{T\rightarrow\infty}N_{T}\geq N_{0}$ then 
\[
\sup_{T\geq T_{0}}\sup_{(\theta,\epsilon,y_{t})\in\mathring{\Theta}\times\Xi\times\mathsf{Y}}\frac{\left|\mathbb{C}_{\theta,\epsilon,T}^{\omega}\left[\dot{\ell}_{\theta}(X_{t}\mid y_{t}),\dot{\ell}_{\tilde{\theta}(\epsilon,T)}(X'_{t}\mid y_{t})\right]\right|}{\sqrt{\mathbb{V}_{\theta}^{\omega}\left(\dot{\ell}_{\theta}\big(X_{t}\mid y_{t}\big)\right)\mathbb{V}_{\theta,\epsilon,T}^{\omega}\left(\dot{\ell}_{\tilde{\theta}(\epsilon,T)}\big(X'_{t}\mid y_{t}\big)\right)}}<1,
\]
implying that the first condition of Lemma \ref{lem:UCLT} below holds. 
\end{cor}
\begin{lem}
\label{lem:UCLT}Assume (A\ref{hyp:strong-likelihood-cond}) and let
$\big\{ N_{T}\big\}\in\mathbb{N}$ be such that for some $T_{0}\in\mathbb{N}$
\begin{equation}
\inf_{T\geq T_{0}}\inf_{(\theta,\epsilon,y_{t})\in\mathring{\Theta}\times\Xi\times\mathsf{Y}}\frac{\mathbb{C}_{\theta,\epsilon,T}^{\omega}\left[\dot{\ell}_{\theta}(X_{t}\mid y_{t}),\dot{\ell}_{\tilde{\theta}(\epsilon,T)}(X'_{t}\mid y_{t})\right]}{\sqrt{\mathbb{V}_{\theta}^{\omega}\left(\dot{\ell}_{\theta}\big(X_{t}\mid y_{t}\big)\right)\mathbb{V}_{\theta,\epsilon,T}^{\omega}\left(\dot{\ell}_{\tilde{\theta}(\epsilon,T)}\big(X'_{t}\mid y_{t}\big)\right)}}>-1.\label{eq:cond-covariance}
\end{equation}
Let $S_{\theta,\epsilon,T}^{(1)}(\omega,\xi)$ be as defined in Lemma
\ref{lem:taylor-approximation} and let for any $(\theta,\epsilon)\in\mathring{\Theta}\times\Xi$
and $\omega\in\mathsf{Y}^{\mathbb{N}}$ 
\[
\bar{S}_{\theta,\epsilon,T}^{(1)}(\omega):=\frac{\epsilon}{2\sqrt{T}}\sum_{t=1}^{T}\mathbb{E}_{\theta,\epsilon,T}^{\omega}\big[\dot{\ell}_{\tilde{\theta}(\epsilon,T)}\big(X'_{t}\mid y_{t}\big)\big]
\]
and 
\[
\sigma_{T}^{2}(\theta,\epsilon;\omega):=\frac{\epsilon^{2}}{4T}\sum_{t=1}^{T}\mathbb{V}_{\theta,\epsilon,T}^{\omega}\left\{ \dot{\ell}_{\theta}(X_{t}\mid y_{t})+\dot{\ell}_{\tilde{\theta}(\epsilon,T)}(X'_{t}\mid y_{t})\right\} .
\]
Then for any $\omega\in\mathsf{Y}^{\mathbb{N}}$, 
\[
\lim_{T\rightarrow\infty}\sup_{(\theta,\epsilon,z)\in\mathring{\Theta}\times\Xi\times\mathbb{R}}\left|\mathbb{P}_{\theta,\epsilon,T}^{\omega}\left(\tfrac{S_{\theta,\epsilon,T}^{(1)}(\xi,\omega)-\bar{S}_{\theta,\epsilon,T}^{(1)}(\omega)}{\sigma_{T}(\theta,\epsilon;\omega)}\leq z\right)-\Phi(z)\right|=0,
\]
where $\Phi(\cdot)$ is the cumulative distribution function of $\mathcal{N}(0,1)$.
\end{lem}
\begin{proof}
For $\delta>0$ let 
\[
N_{\theta,\epsilon,T}^{\omega}:=\sum_{t=1}^{T}\mathbb{E}_{\theta,\epsilon,T}^{\omega}\left[\left\vert \dot{\ell}_{\theta}(X_{t}\mid y_{t})+\dot{\ell}_{\tilde{\theta}(\epsilon,T)}\big(X'_{t}\mid y_{t}\big)-\mathbb{E}_{\theta}^{\omega}\big[\dot{\ell}_{\tilde{\theta}(\epsilon,T)}\big(X'_{t}\mid y_{t}\big)\big]\right\vert ^{2+\delta}\right]
\]
and 
\[
D_{\theta,\epsilon,T}^{\omega}:=\left({\textstyle \sum}_{t=1}^{T}\mathbb{V}_{\theta,\epsilon,T}^{\omega}\big(\dot{\ell}_{\theta}\big(X_{t}\mid y_{t}\big)+\dot{\ell}_{\tilde{\theta}(\epsilon,T)}\big(X'_{t}\mid y_{t}\big)\big)\right)^{1+\delta/2}.
\]
Lyapunov's theorem \citep[Theorem 5.7, p. 154]{petrov1995limit} states
that there exists a universal constant $C$ such that for any $T\in\mathbb{N}$
\[
\sup_{z\in\mathbb{R}}\left|\mathbb{P}_{\theta,\epsilon,T}^{\omega}\left(\tfrac{S_{\theta,\epsilon,T}^{(1)}(\xi,\omega)-\bar{S}_{\theta,\epsilon,T}^{(1)}(\omega)}{\sigma_{T}(\theta,\epsilon;\omega)}\leq z\right)-\Phi(z)\right|\leq C\frac{N_{\theta,\epsilon,T}^{\omega}}{D_{\theta,\epsilon,T}^{\omega}}.
\]
In order to establish our uniform result we will check that we have
$(\theta,\epsilon)-$uniform convergence of the upper bound. Clearly
\[
N_{\theta,\epsilon,T}^{\omega}\leq3^{2+\delta}\bar{\ell}^{(1)}T
\]
and 
\[
D_{\theta,\epsilon,T}^{\omega}\geq T^{1+\delta/2}\left\{ \inf_{(\theta,\epsilon,T,y_{t})\in\mathring{\Theta}\times\Xi\times\mathbb{N}\times\mathsf{Y}}\mathbb{V}_{\theta,\epsilon,T}^{\omega}\left(\dot{\ell}_{\theta}\big(X_{t}\mid y_{t}\big)+\dot{\ell}_{\tilde{\theta}(\epsilon,T)}\big(X'_{t}\mid y_{t}\big)\right)\right\} ^{1+\delta/2}.
\]
If $\inf_{(\theta,\epsilon,T,y_{t})\in\mathring{\Theta}\times\Xi\times\mathbb{N}\times\mathsf{Y}}\mathbb{V}_{\theta,\epsilon,T}^{\omega}\left(\dot{\ell}_{\theta}\big(X_{t}\mid y_{t}\big)+\dot{\ell}_{\tilde{\theta}(\epsilon,T)}\big(X'_{t}\mid y_{t}\big)\right)>0$
then the denominator grows super linearly and we can conclude. We
have 
\begin{multline*}
\mathbb{V}_{\theta,\epsilon,T}^{\omega}\left(\dot{\ell}_{\theta}\big(X_{t}\mid y_{t}\big)+\dot{\ell}_{\tilde{\theta}(\epsilon,T)}\big(X'_{t}\mid y_{t}\big)\right)\\
\geq2\sqrt{\mathbb{V}_{\theta}^{\omega}\left(\dot{\ell}_{\theta}\big(X_{t}\mid y_{t}\big)\right)\mathbb{V}_{\theta,\epsilon,T}^{\omega}\left(\dot{\ell}_{\tilde{\theta}(\epsilon,T)}\big(X'_{t}\mid y_{t}\big)\right)}+2\mathbb{C}_{\theta,\epsilon,T}^{\omega}\left[\dot{\ell}_{\theta}(X_{t}\mid y_{t}),\dot{\ell}_{\tilde{\theta}(\epsilon,T)}(X'_{t}\mid y_{t})\right]
\end{multline*}
and conclude using the second part of Lemma \ref{lem:boundcovariance},
the assumption in \eqref{eq:cond-covariance} and the fact that by
(A\ref{hyp:strong-likelihood-cond}) we have $\inf_{(\theta,\epsilon,y_{t})\in\Theta\times\Xi\times\mathsf{Y}}\mathbb{V}_{\theta}^{\omega}\left(\dot{\ell}_{\theta}\big(X_{t}\mid y_{t}\big)\right)>0$.
\end{proof}
\begin{rem}
The assumption in \eqref{eq:cond-covariance} will be satisfied whenever
$R_{\theta,y}^{[N]}\big(x,\cdot\big)$ is a positive operator, or
can be checked using Corollary \ref{cor:existenceN0}. 
\end{rem}
We now examine the limits as $T\rightarrow\infty$ of $\bar{S}_{\theta,\epsilon,T}^{(1)}(\omega)$
and $\sigma_{T}^{2}(\theta,\epsilon;\omega)$ under the distribution
of the observations $P$.

\subsection{{\normalsize{}{}Limit of the expectations in the conditional CLT}}
\begin{lem}
\label{lem:limitexpectationinCLT}Assume (A\ref{hyp:strong-likelihood-cond}).
Then with the notation from Lemma \ref{lem:UCLT}, for any $\varepsilon_{0}>0$
there exists $N_{0}\in\mathbb{N}$ such that

\[
\sup_{T\geq1}\sup_{(\theta,\epsilon,\omega)\in\mathring{\Theta}\times\Xi\times\mathsf{Y}^{\mathbb{N}}}\left|\bar{S}_{\theta,\epsilon,T}^{(1)}(\omega)\right|\mathbb{I}\{N_{T}\geq N_{0}\}\leq\varepsilon_{0}.
\]
\end{lem}
\begin{proof}
We apply the result of Lemma \ref{lem:basicLipschitzcontinuities}
and use the fact that here for any $\theta\in\Theta$ and $t\geq1$
we have $\mathbb{E}_{\theta}^{\omega}\big[\dot{\ell}_{\theta}(X_{t}\mid y_{t})\big]=0$.
Therefore for any $t\geq1$ and $(\theta,\epsilon,y_{t})\in\mathring{\Theta}\times\Xi\times\mathsf{Y}$,
\begin{multline*}
\left|\mathbb{E}_{\theta}^{\omega}\left[\dot{\ell}_{\theta}(X{}_{t}\mid y_{t})\right]-\mathbb{E}_{\theta,\epsilon,T}^{\omega}\left[\dot{\ell}_{\tilde{\theta}(\epsilon,T)}(X'_{t}\mid y_{t})\right]\right|\leq\big(\bar{\ell}^{(1)}\big)^{2}\times\frac{|\epsilon|}{\sqrt{T}}\sup_{(\theta,x,y)\Theta\times\mathsf{X}\times\mathsf{Y}}\|R_{\theta,y}^{[N_{T}]}\big(x,\cdot\big)-p_{\theta}(\cdot\mid y)\|_{tv}.
\end{multline*}
Hence 
\[
\left|\bar{S}_{\theta,\epsilon,T}^{(1)}\right|\leq\frac{|\epsilon|}{2}\big(\bar{\ell}^{(1)}\big)^{2}\sup_{(\theta,x,y)\in\Theta\times\mathsf{X}\times\mathsf{Y}}\|R_{\theta,y}^{[N_{T}]}\big(x,\cdot\big)-p_{\theta}(\cdot\mid y)\|_{tv},
\]
and we conclude with the increasing ergodicity of $R_{\theta,y}^{[N]}\big(x,\cdot\big)$
with $N$. 
\end{proof}

\subsection{{\normalsize{}{}Limit of the variances in the conditional CLT}}
\begin{lem}
\label{lem:convergencevariance}Assume (A\ref{hyp:strong-likelihood-cond}).
Then for any $\varepsilon_{0}>0$ there exists $T_{0},N_{0}\in\mathbb{N}$
such that 
\[
\sup_{T\geq T_{0}}\sup_{(\theta,\epsilon)\in\mathring{\Theta}\times\Xi}\left|\sigma_{T}^{2}(\theta,\epsilon;\omega)-\sigma^{2}(\theta,\epsilon)\right|\mathbb{I}\{N_{T}\geq N_{0}\}\leq\varepsilon_{0}\quad P-a.s.
\]
\end{lem}
\begin{proof}
Note that for any $t\geq1$ 
\begin{multline*}
\mathbb{V}_{\theta,\epsilon,T}^{\omega}\left(\dot{\ell}_{\theta}\big(X_{t}\mid y_{t}\big)+\dot{\ell}_{\tilde{\theta}(\epsilon,T)}\big(X'_{t}\mid y_{t}\big)\right)\\
=\mathbb{V}_{\theta}^{\omega}\left(\dot{\ell}_{\theta}\big(X_{t}\mid y_{t}\big)\right)+\mathbb{V}_{\theta,\epsilon,T}^{\omega}\left(\dot{\ell}_{\tilde{\theta}(\epsilon,T)}\big(X'_{t}\mid y_{t}\big)\right)+2\mathbb{C}_{\theta,\epsilon,T}^{\omega}\left[\dot{\ell}_{\theta}(X_{t}\mid y_{t}),\dot{\ell}_{\tilde{\theta}(\epsilon,T)}(X'_{t}\mid y_{t})\right].
\end{multline*}
From Lemma \ref{lem:boundcovariance} and (A\ref{hyp:strong-likelihood-cond})
there exist $T_{0,1},N_{0}\in\mathbb{N}$ such that for $T\geq T_{0,1}$
and $N_{T}\geq N_{0}$ 
\[
\sup_{(\theta,\epsilon,y_{t})\in\mathring{\Theta}\times\Xi\times\mathsf{Y}}\left|\mathbb{C}_{\theta,\epsilon,T}^{\omega}\big[\dot{\ell}_{\theta}(X_{t}\mid y_{t}),\dot{\ell}_{\tilde{\theta}(\epsilon,T)}(X'_{t}\mid y_{t})\big]\right|+\left|\mathbb{V}_{\theta,\epsilon,T}^{\omega}\left(\dot{\ell}_{\theta}\big(X_{t}\mid y_{t}\big)\right)-\mathbb{V}_{\theta,\epsilon,T}^{\omega}\left(\dot{\ell}_{\tilde{\theta}(\epsilon,T)}\big(X'_{t}\mid y_{t}\big)\right)\right|\leq\varepsilon_{0}/4
\]
Therefore, since $\sup_{(\theta,\epsilon,y_{t})\in\mathring{\Theta}\times\Xi\times\mathsf{Y}}\mathbb{V}_{\theta,\epsilon,T}^{\omega}\left[\dot{\ell}_{\theta}(X_{t}\mid y_{t})+\dot{\ell}_{\tilde{\theta}(\epsilon,T)}(X'_{t}\mid y_{t})\right]<\infty$,
there exists $T_{0,2}\in\mathbb{N}$ such that 
\[
\sup_{T\geq T_{0,2}}\sup_{(\theta,\epsilon,\omega)\in\mathring{\Theta}\times\Xi\times\mathsf{Y}^{\mathbb{N}}}\left|\frac{\epsilon^{2}}{4T}\sum_{t=1}^{T}\mathbb{V}_{\theta,\epsilon,T}^{\omega}\left[\dot{\ell}_{\theta}(X_{t}\mid y_{t})+\dot{\ell}_{\tilde{\theta}(\epsilon,T)}(X'_{t}\mid y_{t})\right]-\frac{\epsilon^{2}}{2T}\sum_{t=1}^{T}\mathbb{V}_{\theta}^{\omega}\left[\dot{\ell}_{\theta}(X_{t}\mid y_{t})\right]\right|\leq\varepsilon_{0}/2.
\]
Finally we show that $P-a.s.$ 
\begin{equation}
\lim_{T\rightarrow\infty}\sup_{(\theta,\epsilon)\in\Theta\times\Xi}\left|\frac{\epsilon^{2}}{2T}\sum_{t=1}^{T}\mathbb{V}_{\theta}^{\omega}\left[\dot{\ell}_{\theta}(X_{t}\mid Y_{t})\right]-\sigma^{2}(\theta,\epsilon)\right|=0.\label{eq:cvasympvariance}
\end{equation}
This is immediate upon noting that by assumption for any $(\theta,\epsilon)\in\Theta\times\Xi$
and $y_{t}\in\mathsf{Y}$ we have 
\begin{align*}
\mathbb{V}_{\theta}^{\omega}\left[\dot{\ell}_{\theta}(X_{t}\mid y_{t})\right] & =-\mathbb{E}_{\theta}^{\omega}\left[\ddot{\ell}_{\theta}(X_{t}\mid y_{t})\right],
\end{align*}
and by applying Lemma \ref{lem:ucvbarS2}, up to a constant factor.
We deduce the existence of $T_{0,3}$ such that the absolute difference
in \eqref{eq:cvasympvariance} is less than $\varepsilon_{0}/2$ for
$T\geq T_{0,3}$. We conclude by choosing $T_{0}=T_{0,1}\vee T_{0,2}\vee T_{0,3}$. 
\end{proof}

\section{{\normalsize{}{}Proof of the main result and discussion\label{sec:ProofMainResult}}}

Before proving the main result we establish four intermediate results
which will allow us to work with the approximation of $\Lambda_{T}(\theta,\epsilon;\omega,\xi)$
given in Corollary \ref{cor:almostsureapproximation}, the U-CLT and
uniform strong law of large numbers established in Lemmata \ref{lem:UCLT}
and \ref{lem:limitexpectationinCLT}.

\subsection{{\normalsize{}{}Preliminary results}}

In order to simplify notation we introduce a parameter $\vartheta\in\varTheta$
(which plays the role of $\theta,\epsilon$) and introduce associated
sequences of random variables $\{A_{n}^{\vartheta},n\in\mathbb{N}\}$
and $\{B_{n}^{\vartheta},n\in\mathbb{N}\}$ (which play the role of
$\{\Lambda_{T}(\theta,\epsilon;\omega,\xi),T\geq1\}$ and its approximation)
defined on the same space and associated to a probability distribution
denoted $\mathbb{P}_{\vartheta}$. We let $\alpha(x):=1\wedge\exp(x)$. 
\begin{lem}
\label{lem:firstapproxEAP}Let $\{A_{n}^{\vartheta},\vartheta\in\varTheta,n\in\mathbb{N}\}$
and $\{B_{n}^{\vartheta},\vartheta\in\varTheta,n\in\mathbb{N}\}$
be two families of random variables defined on a common probability
space. Let $\big\{\vartheta_{n}\sim\mu_{n},n\in\mathbb{N}\big\}$
for a family of probability distributions $\big\{\mu_{n},n\in\mathbb{N}\big\}$
on $\varTheta$ and an associated $\sigma-$algebra, $\varphi:\varTheta\rightarrow[0,1]$
and $\{a_{n},n\in\mathbb{N}\}$ a real valued sequence. Assume that
$\lim_{n\rightarrow\infty}\sup_{\vartheta\in\varTheta}\mathbb{E}_{\vartheta}\big|A_{n}^{\vartheta}-B_{n}^{\vartheta}\big|=0$.
Then 
\[
\lim_{n\rightarrow\infty}\big\{\mathbb{E}\big[\alpha\big(a_{n}+A_{n}^{\vartheta_{n}}\big)\varphi(\vartheta_{n})\big]-\mathbb{E}\big[\alpha\big(a_{n}+B_{n}^{\vartheta_{n}}\big)\varphi(\vartheta_{n})\big]\big\}=0.
\]
\end{lem}
\begin{proof}
$\alpha(x)$ is Lipschitz since $|1\wedge\exp(x)\text{\textminus}1\wedge\exp(y)|=1\wedge|\exp(0\wedge x)\text{\textminus}\exp(0\text{\ensuremath{\wedge}}y)|\text{\ensuremath{\le}}1\wedge|x\text{\textminus}y|$
and the proof is immediate by using the fact that $\varphi$ is bounded. 
\end{proof}
We need the following intermediate result. 
\begin{lem}
\label{lem:EAPinCDF}Let $Z$ be a random variable on some probability
space with cumulative distribution $F$. Then for any $a\in\mathbb{R}$,
\[
\mathbb{E}\big[1\wedge\exp\big(a+Z\big)\big]=1-\int_{-\infty}^{0}F(u-a)\exp(u){\rm d}u.
\]
\end{lem}
\begin{proof}
We have 
\begin{align*}
\mathbb{E}\big[1\wedge\exp\big(a+Z\big)\big] & =\mathbb{E}\big[\exp\big(a+Z\big)\mathbb{I}\{a+Z\leq0\}\big]+\mathbb{E}\big[\mathbb{I}\{a+Z>0\}\big]\\
 & =\mathbb{E}\big[\mathbb{I}\{a+Z\leq0\}\int_{0}^{1}\mathbb{I}\{t<\exp\big(a+Z\big)\}{\rm d}t\big]+1-F(-a)\\
 & =\int_{0}^{1}\mathbb{E}\big[\mathbb{I}\{\log(t)<a+Z\leq0\}\}\big]{\rm d}t+1-F(-a)\\
 & =\int_{0}^{1}\big[F(-a)-F(\log(t)-a)\big]{\rm d}t+1-F(-a)\\
 & =1-\int_{-\infty}^{0}F(u-a)\exp(u){\rm d}u,
\end{align*}
where we have used Tonelli's theorem. 
\end{proof}
We will use the following technical lemma. 
\begin{lem}
\label{lem:perturbcoefCLT}Consider a sequence $\{(m_{n}^{\vartheta},s_{n}^{\vartheta}),n\in\mathbb{N}\}\in(\mathbb{R}\times\mathbb{R}_{+})^{\mathbb{N}}$,
$s_{-}$, $s_{+}\in\mathbb{R}_{+}\times\mathbb{R}_{+}$ and $(m^{\vartheta}=-s^{\vartheta}/2,s^{\vartheta})\in\mathbb{R}\times\mathbb{R}_{+}$
such that 
\[
\lim_{n\rightarrow\infty}\sup_{\vartheta\in\varTheta}\big(|m_{n}^{\vartheta}-m^{\vartheta}|+|s_{n}^{\vartheta}-s^{\vartheta}|\big)=0,
\]
\[
0<s_{-}\leq\inf_{(\vartheta,n)\in\varTheta\times\mathbb{N}}s_{n}^{\vartheta}\leq\sup_{(\vartheta,n)\in\varTheta\times\mathbb{N}}s_{n}^{\vartheta}\leq s_{+}<\infty.
\]
Then for any $\{a_{n}\}\in\mathbb{R}^{\mathbb{N}}$, 
\[
\lim_{n\rightarrow\infty}\sup_{\vartheta\in\varTheta,u\in\mathbb{R}}\left|\Phi\left(\tfrac{u-a_{n}^{\vartheta}-m_{n}^{\vartheta}}{\sqrt{s_{n}^{\vartheta}}}\right)-\Phi\left(\tfrac{u-a_{n}^{\vartheta}-m^{\vartheta}}{\sqrt{s^{\vartheta}}}\right)\right|=0.
\]
\end{lem}
\begin{proof}
We exploit the mean value theorem and with $\phi(\cdot)$ the probability
density of a standard normal distribution the fact that 
\[
\sup_{s\in[s_{n}^{\vartheta}\wedge s^{\vartheta},s_{n}^{\vartheta}\vee s^{\vartheta}]}\phi\big(\tfrac{u-a_{n}^{\vartheta}-m^{\vartheta}}{\sqrt{s}}\big)\leq\phi\big(\tfrac{u-a_{n}^{\vartheta}-m^{\vartheta}}{\sqrt{s_{+}}}\big).
\]
More precisely for any $(u,\vartheta,n)\in\mathbb{R}\times\varTheta\times\mathbb{N}$
\begin{align*}
\left|\Phi\left(\tfrac{u-a_{n}^{\vartheta}-m^{\vartheta}}{\sqrt{s_{n}^{\vartheta}}}\right)-\Phi\left(\tfrac{u-a_{n}^{\vartheta}-m^{\vartheta}}{\sqrt{s^{\vartheta}}}\right)\right| & \leq\frac{1}{2}s_{-}^{-3/2}\phi\big(\tfrac{u-a_{n}^{\vartheta}-m^{\vartheta}}{\sqrt{s_{+}}}\big)|u-a_{n}^{\vartheta}-m^{\vartheta}|\left|s_{n}^{\vartheta}-s^{\vartheta}\right|
\end{align*}
Clearly 
\[
C:=\frac{1}{2}\sqrt{s_{+}}s_{-}^{-3/2}\sup_{z\in\mathbb{R}}|z|\phi(z)<\infty,
\]
from which we deduce that for any $(u,\vartheta,n)\in\mathbb{R}\times\varTheta\times\mathbb{N}$
\[
\left|\Phi\left(\tfrac{u-a_{n}^{\vartheta}-m^{\vartheta}}{\sqrt{s_{n}^{\vartheta}}}\right)-\Phi\left(\tfrac{u-a_{n}^{\vartheta}-m^{\vartheta}}{\sqrt{s^{\vartheta}}}\right)\right|\leq C\left|s_{n}^{\vartheta}-s^{\vartheta}\right|.
\]
We also have for any $(u,\vartheta,n)\in\mathbb{R}\times\varTheta\times\mathbb{N}$
\[
\left|\Phi\left(\tfrac{u-a_{n}^{\vartheta}-m_{n}^{\vartheta}}{\sqrt{s_{n}^{\vartheta}}}\right)-\Phi\left(\tfrac{u-a_{n}^{\vartheta}-m^{\vartheta}}{\sqrt{s_{n}^{\vartheta}}}\right)\right|\leq\frac{1}{\sqrt{2\pi}}\tfrac{1}{\sqrt{s_{-}}}\left|m^{\vartheta}-m_{n}^{\vartheta}\right|.
\]
We therefore deduce that 
\[
\lim_{n\rightarrow\infty}\sup_{\vartheta\in\varTheta,u\in\mathbb{R}}\left|\Phi\left(\tfrac{u-a_{n}^{\vartheta}-m_{n}^{\vartheta}}{\sqrt{s_{n}^{\vartheta}}}\right)-\Phi\left(\tfrac{u-a_{n}^{\vartheta}-m^{\vartheta}}{\sqrt{s^{\vartheta}}}\right)\right|=0.
\]
\end{proof}
We now establish the log-normal approximation we are interested in. 
\begin{prop}
\label{prop:penaltyapproxEAP}Let $\{B_{n}^{\vartheta},\vartheta\in\varTheta,n\in\mathbb{N}\}$
be a family of random variables defined on some probability space.
Let $\big\{\vartheta_{n}\sim\mu_{n},n\in\mathbb{N}\big\}$ for a family
of probability distributions $\big\{\mu_{n},n\in\mathbb{N}\big\}$
on $\varTheta$ and an associated $\sigma-$algebra, $\varphi:\varTheta\rightarrow[0,1]$
and $\{a_{n},n\in\mathbb{N}\}$ a real valued sequence. Assume there
exist a sequence $\{(m_{n}^{\vartheta},s_{n}^{\vartheta}),n\in\mathbb{N}\}\in(\mathbb{R}\times\mathbb{R}_{+})^{\mathbb{N}}$,
$s_{-}$, $s_{+}$and $(m^{\vartheta}=-s^{\vartheta}/2,s^{\vartheta})\in\mathbb{R}\times\mathbb{R}_{+}$
such that 
\[
\lim_{n\rightarrow\infty}\sup_{\vartheta\in\varTheta}\big(|m_{n}^{\vartheta}-m^{\vartheta}|+|s_{n}^{\vartheta}-s^{\vartheta}|\big)=0,
\]
\[
0<s_{-}\leq\inf_{(\vartheta,n)\in\varTheta\times\mathbb{N}}s_{n}^{\vartheta}\leq\sup_{(\vartheta,n)\in\varTheta\times\mathbb{N}}s_{n}^{\vartheta}\leq s_{+}<\infty
\]
and 
\[
\lim_{n\rightarrow\infty}\sup_{z\in\mathbb{R},\vartheta\in\varTheta}\big|\mathbb{P}_{\vartheta}\big(\tfrac{B_{n}^{\vartheta}-m_{n}^{\vartheta}}{\sqrt{s_{n}^{\vartheta}}}\leq z\big)-\Phi(z)\big|=0.
\]
Then with $B^{\vartheta}\sim\mathcal{N}(-s^{\vartheta}/2,s^{\vartheta})$
\[
\lim_{n\rightarrow\infty}\big\{\mathbb{E}\big[\alpha\big(a_{n}^{\vartheta_{n}}+B_{n}^{\vartheta_{n}}\big)\varphi(\vartheta_{n})\big]-\mathbb{E}\big[\alpha\big(a_{n}^{\vartheta_{n}}+B^{\vartheta_{n}}\big)\varphi(\vartheta_{n})\big]\big\}=0.
\]
\end{prop}
\begin{proof}
We consider the difference and with $F_{\vartheta,n}\big(z\big):=\mathbb{P}_{\vartheta}\left(\tfrac{B_{n}^{\vartheta}-m_{n}^{\vartheta}}{\sqrt{s_{n}^{\vartheta}}}\leq z\right)$
we obtain with Lemma \ref{lem:EAPinCDF} 
\begin{align*}
\mathbb{E}\left\{ \big[\alpha\big(a_{n}^{\vartheta_{n}}+B^{\vartheta_{n}}\big)-\alpha\big(a_{n}^{\vartheta_{n}}+B_{n}^{\vartheta_{n}}\big)\big]\varphi(\vartheta_{n})\right\} \\
=\int_{\varTheta}\mu_{n}({\rm d}\vartheta) & \varphi(\vartheta)\int_{-\infty}^{0}\big[\mathbb{P}_{\vartheta}\big(B_{n}^{\vartheta}\leq u-a_{n}^{\vartheta}\big)-\mathbb{P}_{\vartheta}\big(B^{\vartheta}\leq u-a_{n}^{\vartheta}\big)\big]\exp(u){\rm d}u\\
=\int_{\varTheta}\mu_{n}({\rm d}\vartheta) & \varphi(\vartheta)\int_{-\infty}^{0}\big[F_{\vartheta,n}\left(\tfrac{u-a_{n}^{\vartheta}-m_{n}^{\vartheta}}{\sqrt{s_{n}^{\vartheta}}}\right)-\Phi\left(\tfrac{u-a_{n}^{\vartheta}-m_{n}^{\vartheta}}{\sqrt{s_{n}^{\vartheta}}}\right)\big]\exp(u){\rm d}u\\
\quad+\int_{\varTheta} & \mu_{n}({\rm d}\vartheta)\varphi(\vartheta)\int_{-\infty}^{0}\big[\Phi\left(\tfrac{u-a_{n}^{\vartheta}-m_{n}^{\vartheta}}{\sqrt{s_{n}^{\vartheta}}}\right)-\Phi\left(\tfrac{u-a_{n}^{\vartheta}-m^{\vartheta}}{\sqrt{s^{\vartheta}}}\right)\big]\exp(u){\rm d}u.
\end{align*}
The first term vanishes from the assumed U-CLT and the second from
Lemma \ref{lem:perturbcoefCLT}. 
\end{proof}

\subsection{{\normalsize{}{}Proof of the main result}}

The proof of the main result is a slight modification of the proposition
below relying on Lemma \ref{lem:perturbcoefCLT} and the fact that
$\lim_{T\rightarrow\infty}\sup_{(\theta,\epsilon)\in\Theta\times\Xi}|\varsigma_{T}^{2}(\theta,\epsilon)-\sigma^{2}(\theta,\epsilon)|=0$
from Lemma \ref{lem:basicLipschitzcontinuities}, where we remind
the reader that $\varsigma_{T}^{2}(\theta,\epsilon):=\sigma^{2}(\tilde{\theta}(\epsilon,T),\epsilon)$
and use the notation from Theorem \ref{thm:main}. 
\begin{prop}
\label{prop:nearlymainresult}Assume (A\ref{hyp:strong-likelihood-cond}).
Then $P-$a.s., for any $\varepsilon_{0}>0$ there exist $T_{0},N_{0}\in\mathbb{N}$
such that for any $T\geq T_{0}$ and any sequence $\big\{ N_{T}\big\}\in\mathbb{N}^{\mathbb{N}}$
such that $N_{T}\geq N_{0}$ for $T\geq T_{0}$ 
\[
\sup_{T\geq T_{0}}\left|\mathbb{E}_{T}^{\omega}\left[\min\{1,\tilde{r}_{T}(\theta,\epsilon;\omega,\xi)\}\right]-\check{\mathbb{E}}_{T}^{\omega}\left[\min\{1,r_{T}(\theta,\epsilon;\omega)\exp(Z)\}\right]\right|\leq\varepsilon_{0},
\]
and 
\[
\sup_{T\geq T_{0}}\left|\mathbb{E}_{T}^{\omega}\left[\min\{1,\tilde{r}_{T}(\theta,\epsilon;\omega,\xi)\}\epsilon^{2}\right]-\check{\mathbb{E}}_{T}^{\omega}\left[\min\{1,r_{T}(\theta,\epsilon;\omega)\exp(Z)\}\epsilon^{2}\right]\right|\leq\varepsilon_{0}
\]
where 
\[
Z\mid(\theta,\epsilon,\omega)\sim\mathcal{N}\left(-\frac{\sigma^{2}(\theta,\epsilon)}{2},\sigma^{2}(\theta,\epsilon)\right).
\]
\end{prop}
\begin{proof}[Proof of Proposition \ref{prop:nearlymainresult}]
First note that by assumption $\Theta\smallsetminus\mathring{\Theta}$
has posterior probability zero and that we can ignore the terms such
that $r_{T}(\theta,\epsilon;\omega)=0$ in the expectations involved
(and we therefore assume implicitly below the presence of an indicator
of the event $r_{T}(\theta,\epsilon;\omega)\neq0$ in order to keep
notation simple). Choose $\varepsilon'_{0}>0$. From Corollary \ref{cor:almostsureapproximation},
$P-a.s.$ we have 
\[
\lim_{T\rightarrow\infty}\sup_{(N_{T},\theta,\epsilon)\in\mathbb{N}\times\mathring{\Theta}\times\Xi}\mathbb{E}_{\theta,\epsilon,T}^{\omega}\big|\Lambda_{T}(\theta,\epsilon;\omega,\xi)-S_{\theta,\epsilon,T}^{(1)}(\omega,\xi)+\sigma^{2}(\theta,\epsilon)/2\big|=0.
\]
Therefore we can apply Lemma \ref{lem:firstapproxEAP} for these realisations
of the observations $\omega$ and show that for some $T_{0,1}\in\mathbb{N}$
and $T\geq T_{0,1}$ 
\[
\sup_{(N_{T},\theta,\epsilon)\in\mathbb{N}\times\mathring{\Theta}\times\Xi}\left|\mathbb{E}_{\theta,\epsilon,T}^{\omega}\left[\min\{1,\tilde{r}_{T}(\theta,\epsilon;\omega,\xi)\}-\min\{1,r_{T}(\theta,\epsilon;\omega)\exp(S_{\theta,\epsilon,T}^{(1)}(\omega,\xi)-\sigma^{2}(\theta,\epsilon)/2)\}\right]\right|\leq\varepsilon'_{0}/3.
\]
Let $N_{0,2}\in\mathbb{N}$ be as in Corollary \ref{cor:existenceN0}.
Then from Lemma \ref{lem:UCLT} for any $\big\{ N_{T}\big\}$ such
that $\liminf_{T\rightarrow\infty}N_{T}\geq N_{0,2}$ we have the
existence of $T_{0,2}\in\mathbb{N}$ such that for any $\omega\in\mathsf{Y}^{\mathbb{N}}$
\[
\sup_{T\geq T_{0,2}}\sup_{(\theta,\epsilon,z)\in\mathring{\Theta}\times\Xi\times\mathbb{R}}\left|\mathbb{P}_{\theta,\epsilon,T}^{\omega}\left(\tfrac{S_{\theta,\epsilon,T}^{(1)}(\xi,\omega)-\bar{S}_{\theta,\epsilon,T}^{(1)}(\omega)}{\sigma_{T}(\theta,\epsilon;\omega)}\leq z\right)-\Phi(z)\right|\leq\varepsilon'_{0}/3.
\]
Let $\lambda_{T}:=\log r_{T}(\theta,\epsilon;\omega)$. From Lemma
\ref{lem:limitexpectationinCLT} and \ref{lem:convergencevariance},
there exist $\alpha_{3},\alpha_{4}>0$ and $N_{0,3},T_{0,3}\in\mathbb{N}$
such that 
\[
\sup_{T\geq1}\sup_{(\theta,\epsilon,\omega)\in\mathring{\Theta}\times\Xi\times\mathsf{Y}^{\mathbb{N}}}\left|\bar{S}_{\theta,\epsilon,T}^{(1)}(\omega)\right|\mathbb{I}\{N_{T}\geq N_{0,3}\}\leq\alpha_{3},
\]
\[
\sup_{T\geq T_{0,3}}\sup_{(\theta,\epsilon)\in\mathring{\Theta}\times\Xi}\left|\sigma_{T}^{2}(\theta,\epsilon;\omega)-\sigma^{2}(\theta,\epsilon)\right|\mathbb{I}\{N_{T}\geq N_{0,3}\}\leq\alpha_{4}\quad P-a.s.
\]
and 
\[
\sup_{(z,\theta,\epsilon,\omega)\in\mathbb{R}\times\mathring{\Theta}\times\Xi\times\mathsf{Y}^{\mathbb{N}}}\left|\Phi\left(\tfrac{z-\lambda_{T}+\sigma^{2}(\theta,\epsilon)/2-\bar{S}_{\theta,\epsilon,T}^{(1)}(\omega)}{\sqrt{\sigma_{T}^{2}(\theta,\epsilon;\omega)}}\right)-\Phi\left(\tfrac{z-\lambda_{T}+\sigma^{2}(\theta,\epsilon)/2}{\sqrt{\sigma^{2}(\theta,\epsilon)}}\right)\right|\leq\varepsilon'_{0}/3.
\]
We can now apply Proposition \ref{prop:penaltyapproxEAP} on the intersection
of the sets of \foreignlanguage{british}{realisations} of the observations
above for $\varepsilon_{0}'=\varepsilon_{0}$ and $\varepsilon'_{0}=\varepsilon_{0}/\sup_{\epsilon\in\Xi}\epsilon^{2}$,
$T_{0}=T_{0,1}\vee T_{0,2}\vee T_{0,3}$ and $N_{0}=N_{0,2}\vee N_{0,3}$.
The two statements of the proposition follow by application of the
tower property of the expectation.
\end{proof}

\subsection{{\normalsize{}{}Discussion of the assumptions\label{subsec:DiscussionAssumptions}}}

We here briefly discuss how restrictive our assumptions are. It should
be clear that the following conditions are mild or can be easily lifted: 
\begin{itemize}
\item we have $\Theta\subset\mathbb{R}$ in order to avoid unnecessary technicalities
inherent to the multivariate scenario. It should be clear from the
proof that our result also holds in the multivariate scenario, 
\item the convexity of $\Theta$ is not a requirement, but here simply ensures
that for any $\theta,\theta'\in\Theta$ then $(\theta+\theta')/2\in\Theta$.
More general intermediate points could be considered in non-convex
scenarios, 
\item the differentiability conditions are satisfied if $\mu_{\theta}(x)$
and $g_{\theta}\left(y\mid x\right)$ are three times differentiable
w.r.t. $\theta$ and do not represent a significant restriction. Lipchitz
continuity of the second derivative could replace the existence of
the third derivative. 
\end{itemize}
The more restrictive conditions are, at various degrees, related to
the existence of bounds uniform in $\theta,\epsilon,\omega$ or $\xi$,
implying in particular in practice that $\mathsf{X}$ and $\mathsf{Y}$
are ``bounded'' ($\Theta$ and $\Xi$ are assumed compact, the latter
not being a serious restriction). Inspection of the proof however
suggests that these conditions can be relaxed and the arguments adapted,
albeit at the expense of significant technical complications. Our
first main point is that our proof ignores the fact that the sequence
of posterior distributions with densities $\{\pi_{T}(\theta;\omega);T\geq1\}$
will, under standard assumptions ensuring that a Bernstein-von Mises
result holds, concentrate on a particular value $\theta^{*}\in\Theta$
of the parameter \citet{kleijn2012}. This suggests that uniformity
in a neighbourhood of $\theta^{*}$ should be sufficient (allowing
one, for example, to relax (A\ref{hyp:strong-likelihood-cond})-(\ref{enu:variancepositiveforallthetas}))
and that the control of terms of the form $\mathbb{E}_{\theta,\epsilon,T}^{\omega}\big[\phi(X_{t},y_{t})\big]$
required in our proof may be achieved through establishing explicit
bounds in $\theta$ and $y_{t}$ which can then be controlled via
concentration on the one hand, and the existence of moments of the
observations on the other hand.
\end{document}